\documentclass[11pt]{article}

\usepackage{cite}
\usepackage{graphicx}
\usepackage{textcomp}
\usepackage{amsmath}
\usepackage{amssymb}
\usepackage{graphicx}
\usepackage{float}
\usepackage{caption}
\usepackage{amsthm}
\usepackage{color, xcolor}
\definecolor{green}{rgb}{0,0.5977,0}
\usepackage[linesnumbered,algoruled,boxed,lined]{algorithm2e}

\newcommand{\zz}{\mathbb Z}

\newcommand{\abs}[1]{\left|{#1}\right|}

\newcommand{\suchthat}{\ | \ }

\newcommand{\tth}{^\text{th}}
\newcommand{\genseq}[3]{{#1}_1 #3 {#1}_2 #3 \dots #3 {#1}_{#2}}
\newcommand{\seq}[2]{\genseq{#1}{#2}{,}}

\newcommand{\txt}[1]{\text{#1}}

\newcommand{\stext}[1]{\ \ \ \ \ \text{(#1)}}

\newcommand{\push}{\\ & \ \ \ \ \ \ \ \ \ \ }

\makeatletter 
\g@addto@macro{\@algocf@init}{\SetKwInOut{Parameter}{Parameters}} 
\makeatother

\usepackage{slashbox}
\usepackage{pifont}
\newcommand{\FOpar}{$\textup{FO}[\oplus]$\xspace}
\newcommand{\LFPpar}{$\textup{LFP}[\oplus]$\xspace}
\newcommand{\cmark}{{\color{green}\ding{51}}}
\newcommand{\xmark}{{\color{red}\ding{55}}}
\newcommand{\cmarkg}{{\color{green}\thead{$\sigma \geqC \tau$ or\\$\sigma$ not all-unary}}}
\newcommand{\xmarkg}{{\color{red}$\sigma \geqS \tau$}}
\usepackage{makecell}
\usepackage{enumitem}

\usepackage{hyperref}
\hypersetup{
    colorlinks = true,
    linkcolor={blue!70!black},
    citecolor={blue!70!black},
    urlcolor={blue!70!black}
}
\usepackage{cleveref}
\usepackage{authblk}
\usepackage{thm-restate}
\newtheorem{theorem}{Theorem}
\newtheorem{corollary}[theorem]{Corollary}
\newtheorem{lemma}[theorem]{Lemma}

\newtheorem{observation}[theorem]{Observation}

\newtheorem{conjecture}[theorem]{Conjecture}
\theoremstyle{definition}
\newtheorem{definition}[theorem]{Definition}

\newcommand{\EA}{\textnormal{EA}}
\newcommand{\types}{\textnormal{types}}
\newcommand{\idx}{\textnormal{index}}

\def\N{\mathbb \zz_{\geq 1}}
\def\phi{\varphi}
\newcommand{\pattern}{\textsc{pattern}}
\newcommand{\struc}[1]{\textup{Str}[{#1}]}
\newcommand{\ostruc}[1]{\textup{Str}_{\leq}[{#1}]}
\newcommand{\negl}{\textnormal{negl}}
\newcommand{\geqC}{\succeq_L}
\newcommand{\geqS}{\succeq_S}

\newcommand{\lC}{\prec_L}

\newcommand\bu{\bar u}

\newcommand\bx{\bar x}
\newcommand\by{\bar y}
\newcommand\bz{\bar z}
\newcommand{\eqrel}[1]{\txt{eq}\left(#1\right)}

\title{Pseudorandom Finite Models}
\author[1]{Jan Dreier}
\author[2]{Jamie Tucker-Foltz}
\affil[1]{TU Wien. \texttt{dreier@ac.tuwien.ac.at}}
\affil[2]{Harvard University. \texttt{jtuckerfoltz@gmail.com}}
\begin{document}
\maketitle

\begin{abstract}
    We study pseudorandomness and pseudorandom generators from the perspective
    of logical definability. 
    Building on results from ordinary derandomization and
    finite model theory, we show that it is possible to deterministically
    construct, in polynomial time, graphs and relational structures that are statistically
    indistinguishable from random structures by any sentence of first order or
    least fixed point logics. 
    This raises the question of whether such
    constructions can be implemented via logical transductions from simpler
    structures with less entropy. In other words, can logical formulas be
    pseudorandom generators? We provide a complete classification of when this
    is possible for first order logic, fixed point logic, and fixed point logic
    with parity, and provide partial results and conjectures for first order
    logic with parity.
\end{abstract}
\medskip

\section{Introduction}\label{secIntro}

The worst-case tractability of decision problems has found an elegant analog in the realm of finite model theory, where hardness is measured not in terms of the computational resources required to solve a given problem, but the logical constructs needed to define it. From the pioneering work of Fagin \cite{FaginsTheorem} to the recent ``quest to capture polynomial time'' \cite{CapturingP}, a wide variety of complexity classes for decision problems have been described in terms of logics \cite{DCBook}. Recent work has begun to extend descriptive complexity into the realm of approximation as well \cite{DefinableInapproximabilityJournal, InapproximabilityUGInFPC}. However, very little has been said about \emph{average case} computational problems from a logical perspective.

A foundational conjecture of cryptography is the existence of cryptographically secure pseudorandom generators (PRGs): Deterministic polynomial-time algorithms that map random inputs to longer pseudorandom outputs in such a way
that no polynomial-time algorithm can distinguish the output from real randomness.
This may be thought of as a game between a \emph{generator} who builds the pseudorandom output
and an \emph{adversary} who with probability \(\frac12\) either sees the generator's output or a truly random string
and has to distinguish the two with probability significantly better than chance. 
The existence of PRGs is equivalent to the existence of one-way functions, which is a standard and widely-believed conjecture in cryptography (see Section \ref{subCryptoBackground}).
In this paper, we study the ability of logical formulas to act as generator or adversary.
The logics we consider are first order (FO) logic and its well-studied extensions with operators 
for computing least fixed points (LFP) and parities (\FOpar and \LFPpar).

Instead of considering random strings, we measure randomness via random graphs and relational structures. This allows us to build off existing results about these logics and study intriguing questions specific to \emph{unordered} structures, meaning those without a binary relation symbol that is always interpreted as a total order.
For varying relational signatures, we consider structures where each tuple appears in a relation independently with probability \(\frac12\).
For a single symmetric binary relation, this yields, for example, the classic Erd\"os-Renyi random graph $G(n, 1/2)$.
The expressive power of logics on such structures is well understood:
the \emph{zero-one laws} for FO \cite{Fagin01Law} and LFP \cite{LFP01Law}
state that, for any sentence \(\phi\),
as the size of the universe $n$ grows, the probability of satisfying the sentence
converges exponentially fast to either \emph{zero} or \emph{one}, depending on \(\phi\).
A variant of this phenomenon holds for \FOpar as well (by \cite{FOPlusParity}, reviewed in detail in Section \ref{secFOPlusParity}),
where on random graphs the probability of being satisfied converges to two limits (one for even $n$ and one for odd $n$).
On the other hand, \LFPpar can canonize random structures with high probability \cite{RandomBinaryCanonization}, implying (via the Immerman-Vardi theorem \cite{ImmermanVardi1, ImmermanVardi2}) that \LFPpar is as powerful as polynomial time computation when it comes to random inputs, so obviously no zero-one law can possibly hold.

In this paper, we study structures that are not random, but
\emph{pseudorandom}. Our formal definition of this concept, presented in
Section~\ref{secModel}, closely mirrors the standard definition from
private-key cryptography.\footnote{To the best of our knowledge, this paper is
the first to investigate cryptographic notions in the setting of logical definability. There is a literature on ``pseudorandom graphs'' but this typically refers to graphs satisfying useful statistical properties of graphs such as subgraph densities and spectral gap. These graphs are not generally pseudorandom in our context. See Krivelevich and Sudakov \cite{PseudorandomGraphs} for a survey.} 
The zero-one law suggests that FO and LFP should admit polynomial-time PRGs: 
To generate pseudorandomness that fools an adversary based on a formula \(\phi\), it suffices to construct a structure that accepts \(\phi\) if and only if the limiting probability is one rather than zero---and nearly all structures satisfy this property.
However, complexity theory is littered with examples of similar tasks (known as ``finding hay in a haystack'' \cite[Chapter 21]{AroraBarak}) that turn out to be quite difficult.

\subsection{Deterministic construction of existentially-closed graphs}

Our first main result, presented in \Cref{secRado}, is a deterministic, polynomial-time construction of graphs and other structures
that FO and LFP cannot distinguish from truly-random structures (\Cref{thm:RadoGenerator}).
We do so by building patterns described by the so-called \emph{extension axioms} (defined in \Cref{secRado}).
These patterns can be found with high probability in random structures and are responsible for forcing the zero-one convergence of FO and LFP sentences.

Over graphs, this can be understood as a finite analogue of an infinite construction: The \emph{Rado graph} is the unique (up to isomorphism) infinite graph satisfying all extension axioms.
Sampling from the infinite Erd\"os-Renyi graph $G(\infty, 1/2)$ yields with probability one the Rado graph.
There are several existing explicit constructions%
\footnote{For example, use the vertex set \(\N\) and connect \(x\) and \(y\) if \(x < y\) and the \(x\tth\) bit of the binary representation of \(y\) is one~\cite{RadoAckermann,Rado1964}.}
and their finite prefixes can be computed in polynomial time. However, these prefixes do \emph{not} satisfy the extension axioms, and hence do not provide pseudorandom structures even against FO adversaries.

It turns out to be a much more challenging task to deterministically construct \emph{finite} analogues of the Rado graph, which are also commonly referred to as \emph{existentially-closed} graphs. Explicit graphs satisfying the extension axioms have been discovered, but they are delicate constructions based on special properties of mathematical structures, such as Paley graphs \cite{PaleyExtensionAxioms}, binary matrices \cite{BlassRossmanQA}, and affine designs \cite{DesignTheory}; see Bonato \cite{ExtensionSurvey} for an overview of these constructions. None of them have obvious generalizations to more complicated finite models with multiple non-symmetric relations of arbitrary arities.\footnote{There is one exception that we are aware of: Even-Zohar, Farber, and Mead \cite{SimplicialComplexes} extend the Paley graph construction to work for hypergraphs using a clever application of Vandermonde determinants. However, it is still unclear how to extend this idea further to work for arbitrary relational signatures.} To this end, we develop a more robust, versatile framework using tools from the literature on derandomization---specifically \emph{perfect hash functions} and \emph{universal sets}. Our approach is quite different from its predecessors and easily handles graphs as well as arbitrary relational structures.

\subsection{Logic vs. logic}
Just as we restrict the adversary to being definable in a given logic,
we may ``level the playing field'' by also restricting the generator to a logic.
To define what it means to generate an output from an input via logic,
we consider the notion of \emph{transductions}~\cite{courcelle2012graph} (formally defined in \Cref{subTransductionsDef}). 
For the purpose of this paper, transductions are tuples of logical formulas that map all structures with a certain signature \(\sigma\)
to structures with the same universe but different signature \(\tau\).
For example, consider an input signature \(\sigma\) with a single binary relation and an output signature \(\tau\) with a single ternary relation \(T\).
Then a transduction \(\Theta\) from \(\sigma\) to \(\tau\) consists of a single \(\sigma\)-formula \(\theta(x,y,z)\) that associates with every input \(\sigma\)-structure \(\mathbb{A}\)
the \(\tau\)-structure \(\Theta(\mathbb A)\) where \(T(x,y,z)\) holds if and only if \(\theta(x,y,z)\) holds in \(\mathbb A\).
We say that $\Theta$ is a \emph{pseudorandom generator from \(\sigma\) to \(\tau\)} if applying $\Theta$ to a random \(\sigma\)-structure 
yields a pseudorandom \(\tau\)-structure.
Thus, in our example, \(\Theta\)~gets a binary \(\sigma\)-structure with \(n^2\) bits of entropy and has to generate a pseudorandom output 
(that the adversary cannot distinguish from a ternary \(\tau\)-structure) with \(n^3\) bits of entropy.
We consider this to be an appropriate model-theoretic analogue of PRGs.

Our second main result, proved in Section \ref{secPRGClassification}, is a complete classification of when such PRGs exist, 
depending on the logics we require the generator and adversary to be in, 
from all possible pairs among $\{\txt{FO}, \txt{LFP}, \txt{\LFPpar}\}$, as well as the input and output signatures.
Table \ref{tab23} lists the results for the special case where $\sigma$ has just one binary relation and $\tau$ has just one ternary relation.

\begin{table}\begin{center}
    \begin{tabular}{|c|c|c|c|}
        \hline
        \backslashbox{$\Theta$}{$\phi$} & FO     & LFP    & \LFPpar                                       \\
        \hline
        FO                              & \xmark & \xmark & \xmark                                        \\
        \hline
        LFP                             & \xmark & \xmark & \xmark                                        \\
        \hline
        \LFPpar                         & \cmark & \cmark & \thead{{\color{blue}$\iff$}{\color{blue}OWFs exist}} \\
        \hline
    \end{tabular}
    \caption{Existence/nonexistence of PRGs from a relational signature with one binary relation to a relational signature with one ternary relation, depending on the logics that the transduction $\Theta$ and adversary $\phi$ are allowed to be defined in. For the bottom-right entry, existence is equivalent to the standard conjecture in private-key cryptography that one-way functions exist (also equivalent to Conjecture \ref{cnjOWF}). A general classification for all possible pairs of signatures is given in Table \ref{tabgeneral}.}
    \label{tab23}
\end{center}\end{table}

Note that, in this model-theoretic setting, entropy comes in different ``shapes.''
Consider PRGs from structures with a single ternary relation to structures with three binary relations.
Ordinarily, it is trivial to reshape and truncate an input containing \(n^3\) independent random bits into an output with \(3 n^2\) bits,
but it is not immediately clear whether logics can accomplish this task as well. 
We come up with a measure of ``entropy'' that completely determines between which signatures pseudorandom generators exist.
Surprisingly, this notion depends on the pairs of logics under consideration, and does not always define a total order on the signatures.

\subsection{First order with parity}
It is notable that none of the logics we considered so far admit a PRG that is unconditionally secure against adversaries from the same logic: FO and LFP are too weak while \LFPpar is too strong.
To find an interesting middle ground, we turn to \FOpar, which does not come close to capturing polynomial time yet can still express interesting properties like ``there are an odd number of triangles''
which hold in random graphs with limiting probability~\(\frac12\).
We consider the \FOpar-transduction that creates a \(t\)-hypergraph from a normal graph
by adding each hyperedge \(\{v_1,\dots,v_t\}\) if and only if 
there are an odd number of vertices that are adjacent to each of \(v_1,\dots,v_t\).
We show in \Cref{secFOPlusParity} that this transduction is a pseudorandom generator that is secure against FO and LFP.
This is done by proving that the transduced hypergraph satisfies with high probability the corresponding hypergraph-extension axioms.
Hence, our transduction takes a random graph with \({n \choose 2}\) bits of entropy
and yields a hypergraph that an FO or LFP adversary 
cannot distinguish from a real \(t\)-hypergraph with \({n \choose t}\) bits of entropy.
We conjecture that this specific transduction is also secure against \FOpar,
and that therefore \FOpar admits a PRG that is secure against itself.

\section{Preliminaries}\label{secModel}

For any nonnegative integer $n$, let $[n] := \{1, 2, 3, \dots, n\}$. By $\log(n)$ we mean the logarithm with base 2. We say a function \(f(n)\) is \emph{negligible} if, for any polynomial function $p$, for all sufficiently large $n$, \(|f(n)| \leq 1/p(n)\). We write \(f(n) = \textnormal{negl}(n)\) to indicate this. For any finite set $S$, we denote sampling an element $s$ uniformly from $S$ by $s \sim S$, and sampling according to a distribution $\mathcal{D}$ over $S$ by $s \sim \mathcal{D}$. For any nonnegative integer $k$, ${S \choose k}$ is the set of all $k$-element subsets of $S$.

\subsection{Models and logics}\label{subModelsAndLogics}

In this paper we study two kinds of random structures.
First, \emph{relational structures}, which we define with respect to a \emph{relational signature} $\sigma = \langle \seq{R}{t} \rangle$, where each relation symbol $R_i$ has an associated \emph{arity} $a_i$. A \emph{$\sigma$-structure} $\mathbb{A}$ is a finite universe $A$ together with a sequence of relations $\langle R_1^\mathbb{A}, R_2^\mathbb{A}, \dots, R_t^\mathbb{A} \rangle$, where each $R_i^\mathbb{A}$ is an $a_i$-ary relation over $A$, that is, a subset of $A^{a_i}$. We identify \(\mathbb A\) with its universe \(A\) and thus write, for example, \(a \in \mathbb A\) instead of \(a \in A\).
We denote the set of all $\sigma$-structures by $\struc{\sigma}$ and the set of all $n$-element $\sigma$-structures by $\struc{\sigma, n}$.
Uniformly sampling from $\struc{\sigma, n}$ yields a \emph{random $\sigma$-structure}, where every ordered $a_i$-tuple of (possibly non-distinct) elements of the universe is in relation $R_i$ independently with probability~$\frac12$.
    
Second, \emph{$t$-hypergraphs} are relational structures consisting of a single symmetric, $t$-ary hyperedge relation where every \(t\)-tuple consists of \(t\) distinct elements. A $2$-hypergraph is simply a \emph{graph}. 
The Erd\"os-Renyi random graph $G(n, 1/2)$ is defined as the uniform distribution over all \(n\)-vertex graphs, that is, the distribution where each edge appears independently with probability \(1/2\).
Similarly, \(G_t(n,1/2)\) denotes the uniform distribution over all \(n\)-element \(t\)-hypergraphs, where each set of \(t\) distinct vertices forms a hyperedge independently with probability \(1/2\).

For any class of graphs or relational structures $\mathcal{P}$ and logic $L$, we say that $\mathcal{P}$ is \emph{definable} in $L$ if there is some sentence $\phi$ in $L$ such that $\mathcal{P}$ consists precisely of all graphs/structures that model $\phi$.
We assume the reader is familiar with first order logic (FO). We denote the extension of FO with a least fixed-point operator as LFP. It is not necessary for us to go into detail about LFP except to note the following property it has. The \emph{Immerman-Vardi theorem} states that, in the presence of a relation symbol that is interpreted as a total order over the universe, the expressive power of LFP is equivalent to polynomial time computation:
\begin{theorem}[Immerman \cite{ImmermanVardi1}, Vardi \cite{ImmermanVardi2}]\label{thmImmermanVardi}
    Let $\sigma$ be a signature containing a binary relation symbol $\leq$, and let $\mathcal{P}$ be a property of $\sigma$-structures such that, for any $\sigma$-structure $\mathbb{A}$ with property $\mathcal{P}$, $\leq^{\mathbb{A}}$ is a total order over the universe. Then $\mathcal{P}$ is decidable in polynomial time if and only if it is definable in \textnormal{LFP}.
\end{theorem}
It is often easier to think of operations on finite structures as algorithms, with the implicit understanding that properties computed by the algorithm can be expressed in a given logic. In this paper, we frequently do not go into detail about how such a translation works. The Immerman-Vardi Theorem assures us that, as long as we have defined a total order over the universe somehow, we can
superimpose this order onto the structure and apply \Cref{thmImmermanVardi} to this ``implicit'' structure.
This lets us successfully translate any steps of the algorithm that are polynomial-time into LFP.

\subsection{Transductions}\label{subTransductionsDef}

A \emph{transduction} is a way of defining one model from another, possibly of a different signature. In the context of relational structures, we define a transduction as follows. Let $\sigma$ and $\tau$ be signatures, where $\tau = \langle \seq{R}{t} \rangle$, and $R_i$ has arity $a_i$. Then a \emph{transduction from $\sigma$ to $\tau$} (also known as an \emph{interpretation of $\tau$ in $\sigma$}) is a sequence of \mbox{$\sigma$-formulas} $\Theta = (\seq{\theta}{t})$, where each $\theta_i$ has $a_i$ free variables $\seq{x}{a_i}$. Overloading notation, we think of $\Theta: \struc{\sigma} \to \struc{\tau}$ as a function mapping \mbox{$\sigma$-structures} to \mbox{$\tau$-structures}, where, given any $\sigma$-structure $\mathbb{A}$, we define $\Theta(\mathbb{A})$ to be the $\tau$-structure over the same universe as that of $\mathbb{A}$, where relation $R_i^{\Theta(\mathbb{A})}$ holds on a tuple $(\seq{x}{a_i})$ if and only if $\mathbb{A} \models \theta_i(\seq{x}{a_i})$. In the contexts of graphs and hypergraphs, we denote by \(\theta(G)\) the transduction where a single relation based on formula \(\theta\) is created.

\subsection{The parity operator}\label{subParityOpDef}

We may augment any logic $L$ with a \emph{parity} operator to form the new logic $L[\oplus]$. The expression \(\oplus y~\phi\) asks whether there is an odd number of elements \(y\) in the universe satisfying the subformula \(\phi\).

This operator adds considerable power to LFP on random structures. In particular, it makes it possible to define a total order over the universe with probability $1 - \negl(n)$ \cite{RandomBinaryCanonization}. In Appendix \ref{appRandomCanonizationProof}, we prove the following variant of this result, which says that we can almost always use a relation of arity $k \geq 2$ to simultaneously order the universe and extract $\Omega(n^k)$ independent random bits.
The Immerman-Vardi theorem lets us then apply polynomial-time computations to these bits.

\begin{restatable}{lemma}{lemRandomCanonization}\label{lemRandomCanonization}
    Let $\sigma$ be a relational signature containing a relation symbol $R$ of arity $k \geq 2$, and let $\tau$ be a relational signature containing symbols ``\(\le\)'' of arity 2 and $Q$ of arity~$k$. There is an \LFPpar transduction $\Gamma: \struc{\sigma} \to \struc{\tau}$ such that, with probability $1 - \negl(n)$ over a random $\sigma$-structure $\mathbb{A}$,
    \begin{enumerate}[label={(\roman*)}]
        \item\label{itmRandomCanonizationOrder} $\leq^{\Gamma(\mathbb{A})}$ defines a total order over the universe, and
        \item\label{itmRandomCanonizationLargeExtraction} $\abs{Q^{\Gamma(\mathbb{A})}} = \Omega(n^k)$.
    \end{enumerate}
    Furthermore, it is always the case that, for $k$-tuples $(\seq{x}{k}) \in Q^{\Gamma(\mathbb{A})}$, even after conditioning on any realization of $\leq^{\Gamma(\mathbb{A})}$ and $Q^{\Gamma(\mathbb{A})}$, the probabilities that $(\seq{x}{k}) \in R^{\mathbb{A}}$ are all $\frac12$ and pairwise independent across $(\seq{x}{k}) \in Q^{\Gamma(\mathbb{A})}$.
\end{restatable}

\subsection{Ordinary pseudorandomness and cryptography}\label{subCryptoBackground}

We refer the reader to the textbook by Boneh and Shoup \cite{BonehShoup} for a comprehensive discussion of the cryptographic notions used in this paper. Here we only provide a brief overview.

A \emph{pseudorandom generator} is a collection of functions $f_n: \{0, 1\}^n \to \{0, 1\}^{\ell(n)}$ such that, for any collection of polynomial-sized circuits $A_n: \{0, 1\}^{\ell(n)} \to \{0, 1\}$,
\begin{equation*}
    \abs{\Pr_{s \sim \{0, 1\}^n}[A_n(f_n(s)) = 1] - \Pr_{t \sim \{0, 1\}^{\ell(n)}}[A_n(t) = 1]} = \negl(n).
\end{equation*}
We may alternatively think of the collection of functions $A_n$ as a single
Turing machine with access to ``advice'' bits (storing for each \(n\) a circuit)
that runs in polynomial time (to simulate the circuit). The advice bits may depend on $n$ in an arbitrary (not necessarily computable) way.

It is believed that there exist pseudorandom generators that extend \(n\) random bits to \(\ell(n) > n\) pseudorandom bits.
\begin{conjecture}\label{cnjOWF}
    There exists a polynomial-time pseudorandom generator $f_n: \{0, 1\}^n \to \{0, 1\}^{n + 1}$.
\end{conjecture}
This is known to be equivalent to a list of other conjectures in cryptography, including:
\begin{itemize}
    \item The existence of one-way functions (OWFs), which are functions that are easy to compute but hard to invert given the output of a random input.
    \item The existence of a computationally secure private-key encryption scheme between two parties with a short pre-distributed key.
    \item Even stronger security guarantees against chosen-plaintext or chosen-ciphertext attacks, in which we assume the adversary has query access to encrypt/decrypt other messages.
\end{itemize}
The hardest part of the equivalence is due to the celebrated result of H\r{a}stad, Impagliazzo, Levin, and Luby \cite{OwfToPrg}, which shows how to construct a PRG from an arbitrary OWF. Since OWFs are more fundamental, simple objects, Conjecture \ref{cnjOWF} is often referred to as the \emph{OWF Conjecture}. The conjecture implies P~\(\neq\)~NP.

Note that, in the definitions of each of these cryptographic objects, we must always allow the adversary to be nonuniform, otherwise some of the reductions do not go through properly. For example, one classic reduction, known as the \emph{length-extension theorem},\footnote{See Boneh and Shoup \cite[Theorem 3.3]{BonehShoup}.} takes a PRG $f_n: \{0, 1\}^n \to \{0, 1\}^{n + 1}$ and iterates it a polynomial $p(n)$ number of times to produce a PRG $g_n: \{0, 1\}^n \to \{0, 1\}^{p(n)}$. After each iteration we output a single bit and feed the remaining $n$ bits back into $f_n$ for the next iteration, until we have output $p(n)$ bits. The proof of security uses the \emph{hybrid argument}, which considers distributions $H_0, H_1, H_2, \dots, H_{p(n)}$ over $p(n)$-bit strings where, in $H_i$, the first $i$ bits are chosen uniformly at random, while the remaining bits are generated from iterating $f_n$ as in the definition of $g_n$. It suffices to prove that, for each $i$, no adversary can distinguish strings drawn from $H_i$ from those drawn from $H_{i + 1}$ with non-negligible advantage; from this it follows that no adversary can distinguish $H_0$ from $H_{p(n)}$, which means $g_n$ is secure. This is true because an adversary that distinguishes $H_i$ from $H_{i + 1}$ can be transformed into an adversary that breaks the security of $f_n$. However, this algorithm requires randomness, and the only known way to derandomize it is to encode a good draw of the random bits as nonuniform advice.

In the proof of one of the cases of Theorem \ref{thmFPCPRG}, we require a novel adaptation of the length-extension construction that can be implemented as an \LFPpar-transduction.

\subsection{Model-theoretic pseudorandomness}

Consider a collection of probability distributions \((\mathcal{D}_n)_{n \in \N}\), where \(\mathcal{D}_n\) ranges over structures \(\struc{\tau,n}\).
Now, we may say that this distribution is uniformly pseudorandom for a logic \(L\) if, for every $\tau$-sentence \(\phi \in L\),
\[
    \abs{\Pr_{\mathbb A \sim \mathcal{D}_n}[\mathbb A \models \phi] - \Pr_{\mathbb A \sim \struc{\tau,n}}[\mathbb A \models \phi]} = \negl(n).
\]
While in the ordinary notion of PRGs, the adversary has access to a polynomial number of non-uniform advice bits,
this model-theoretic definition is completely uniform and makes no mention of advice.
It may be seen as a strength of this model that almost all of our results hold without the presence of advice (specifically, all results except for Theorem \ref{thmFPCPRG} part \ref{itmFPCPRGConditionallyExists}). However, to link our results to the foundational conjectures of cryptography, we will also give our logical adversary access to advice.
This is done via disjoint polynomial-size advice structures \((\mathbb X_n)_{n \in \N}\) that may depend non-uniformly on~\(n\).

For a \(\tau\)-structure \(\mathbb A\) and \(\rho\)-structure \(\mathbb X_n\),
we denote by \((\mathbb A, \mathbb X_n)\) the \((\tau \cup \rho)\)-structure obtained by the disjoint union%
\footnote{The universe of \((\mathbb A, \mathbb X_n)\) is the disjoint union of the universes of \(\mathbb A\) and \(\mathbb X_n\), and each relation
from \(\tau \cup\rho\) is interpreted as the disjoint union of the relations of the two structures, where missing relations are considered empty.}
of the two structures.
A distribution \((\mathcal{D}_n)_{n \in \N}\) is \emph{pseudorandom for a logic \(L\)} if, for every advice-signature \(\rho\), 
every sequence of \mbox{\(\rho\)-structures} \((\mathbb X_n)_{n \in \N}\), where \(\mathbb X_n\) has a universe of size \(n\),
and every $(\tau \cup \rho)$-sentence \(\phi \in L\),
\[
    \abs{\Pr_{\mathbb A \sim \mathcal{D}_n}[(\mathbb A,\mathbb X_n) \models \phi] - \Pr_{\mathbb A \sim \struc{\tau ,n}}[(\mathbb A, \mathbb X_n) \models \phi]} = \negl(n).
\]

Finally, for any pair of logics $L_1$ and $L_2$, an \emph{$(L_1, L_2)$-pseudo\-random generator from $\sigma$ to $\tau$} is an $L_1$ transduction $\Theta: \struc{\sigma} \to \struc{\tau}$ such that the family of distributions obtained by sampling $\mathbb{C}_n \sim \struc{\sigma, n}$ and outputting $\Theta(\mathbb{C}_n)$ is pseudorandom for $L_2$.

\section{Pseudorandom Structures for LFP \\ and Extension Axioms}\label{secRado}

A central property for creating pseudorandom structures with respect to FO or LFP adversary are the so-called \emph{extension axioms}.
For now, we define them for graphs and give the more involved definition for relational structures later.

\begin{definition}
    A graph \(G\) satisfies the \emph{\(k\)-extension axioms \(\EA_k\)} if, for all sets \(S \subseteq V(G)\) of size \(k\) and all
    \(T \subseteq S\), there exists \(v \not \in S\) (called an extension vertex of \((S,T)\)) that is adjacent to every
    vertex in \(T\) and non-adjacent to every vertex in \(S \setminus T\).
\end{definition}

Informally speaking,
the \(k\)-extension axioms
state that any set of \(k\) pebbles in the Ehrenfeucht–Fra\"iss\'e game for (finitary or infinitary) first order logic can be extended in all possible ways by an additional pebble.\footnote{While not necessary for comprehending this paper, we refer the reader to Libkin \cite{FMT} for more background on pebbling games.}
Thus, if two structures satisfy the \(k\)-extension axioms, then Duplicator always has a way to respond regardless of where the pebbles are.
In other words, for any given LFP sentence $\phi$, there exists $k$ such that, for any pair of graphs \(G_1,G_2\) both satisfying the \(k\)-extension axioms \(\EA_k\), we have that \(G_1 \models \phi \Leftrightarrow G_2 \models \phi\).
While this is a standard argument, for completeness, we prove this statement (even in the presence of our additional advice structure) later in~\Cref{lem:EA-determine-qtype}.

An important observation is that a random graph satisfies the extension axioms with high probability.
As we show in \Cref{lem:EA-probs}, for every positive integer \(k\),
    \[
        \Pr_{G \sim G(n, 1/2)}\bigl[ G \not \models \EA_k \bigr] = \negl(n).
    \]

The zero-one laws for FO \cite{Fagin01Law} and LFP \cite{LFP01Law} then follow immediately:
any given sentence with quantifier rank \(k\) either holds in all or no graphs satisfying the \(k\)-extension axioms, and thus on almost all or almost no graphs.
In particular, the infinite random graph \(G(\infty,1/2)\), also called the Rado graph, satisfies the \(k\)-extension axioms for all \(k \in \N\) with probability 1.

In this section, we provide a polynomial-time deterministic construction of \(n\)-element graphs and relational structures
satisfying the \(k\)-extension axioms for \(k = \log(\log(\Theta(n)))\), implying that FO and LFP adversaries cannot distinguish them from truly-random graphs and structures.

\paragraph*{Universal sets and perfect hash functions}
The following definitions provide a natural way to ``extend'' sets of size \(k\) from a linear sized set using only a logarithmic number of elements.
\begin{definition}
    Let \(\mathcal{F}\) be a family of functions from $[n]$ to $[k]$. We say $\mathcal{F}$ is an \emph{$(n,k)$-family of perfect hash functions}
    if for every set $S \subseteq [n]$ of size $k$,
    there exists \(f \in \mathcal{F}\) such that \(\{ f(s) \mid s \in S \} = [k]\).
    A family $\cal U$ of subsets of $[n]$ is an \emph{$(n,k)$-universal set}
    if for every subset $S \subseteq [n]$ of size $k$,
    the family $\{S \cap U \colon U \in \cal U\}$ contains all $2^k$ subsets of~$S$.
\end{definition}

Perfect hash functions and universal sets of small size exist and can be constructed efficiently.
We use the following classic result.

\begin{theorem}[Naor, Shulman, and Srinivasan \cite{naor1995splitters}]\label{thm:hashfunctions}
    There exists a constant $c$ such that, for $n,k \ge 1$, one can construct an $(n,k)$-family of perfect hash functions
    and an $(n,k)$-universal set
    of size
    $2^{ck} \log(n)$ in time $2^{ck} n \log(n)$.
\end{theorem}

\paragraph*{Tournaments}
A \emph{tournament} is an edge orientation of a complete undirected graph.
This means it is a directed graph that contains for every pair of vertices $a,b$ either
the directed edge from $a$ to $b$ or from $b$ to $a$ but not both.
By brute-force, we find a tournament graph of size \(\log(\Theta(n))\) that reserves for each set of size \(k\) 
a vertex that points towards all elements of the set, and use it to assign the extension vertices, as shown in Figure \ref{fig:rado-construction}.

\begin{lemma}\label{lem:tournament}
    Let $k$ be a nonnegative integer.
    There exists a tournament $F$ of size $2^{3k}$
    such that, for every $S \subseteq V(F)$ of size $k$, there is a vertex
    that has a directed edge towards every vertex in $S$.
\end{lemma}
\begin{proof}
    Let us consider a random tournament $F$ on $2^{3k}$ vertices where the orientation of every edge is chosen independently uniformly at random.
    We use the probabilistic method, proving the statement by showing that there is a non-zero probability that $F$ has the stated properties.
    Let us fix a set $S \subseteq V(F)$  of size $k$.
    The probability that a fixed vertex has edges directed towards every vertex in $S$ is precisely $2^{-k}$.
    The probability that there is no vertex with edges directed towards every vertex in $S$ is
    (using \Cref{obs:ProbSingleEvent} in Appendix \ref{apxSimpleFacts}) at most
    \[
        (1 - 2^{-k})^{2^{3k}-k} \le e^{-(2^{3k}-k)/2^{k}}.
    \]
    By the union bound, the probability that, for some set $S$ of vertices of size $k$,
    there is no vertex with edges directed towards every vertex in $S$ is at most
    $$
        {2^{3k} \choose k} \cdot e^{-(2^{3k}-k)/2^{k}}.
    $$
    Basic calculus yields that this term is smaller than \(1\) for all \(k\).
    This means there is a non-zero probability that \(F\) has the stated properties.
\end{proof}

\subsection{Graphs}

We present our construction for graphs first, and afterwards for relational structures.

\begin{theorem}\label{thm:radograph}
    There exists a constant $c$ such that, for every $n$, one can construct
    in time \(O(n^{3})\) an \(n\)-vertex graph satisfying the extension axioms
    $\EA_{k}$ for all \(k \le \lfloor\log(\log(n))/c\rfloor\).
\end{theorem}

Since the extension axioms hold with high probability and determine the truth of LFP sentences,
it immediately follows that LFP cannot distinguish between a truly random graph \(G(n,1/2)\)
and the pseudorandom output of \Cref{thm:radograph}.
We formalize this later in \Cref{thm:RadoGenerator}.

\begin{proof}[Proof of \Cref{thm:radograph}]
    We prove the following claim that implies the theorem: There exists a constant $c$ such that, for every $n,k \in \N$ with $2^{k^c} \le n$,
    one can construct in time $O(2^{2^{ck}} n^2)$ an $n$-vertex graph satisfying \(\EA_k\).

    \begin{figure}
        \begin{center}
            \includegraphics[scale=0.88]{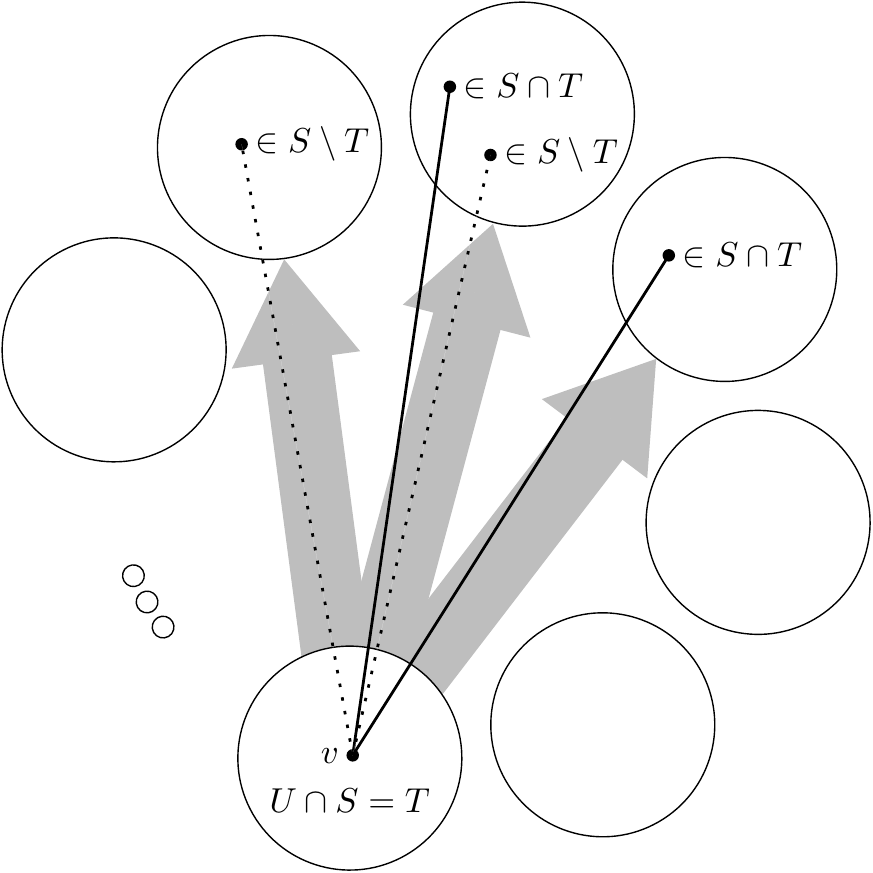}
        \end{center}
        \caption{
            Our finite Rado graph construction. To find an extension vertex to sets \(S\) and \(T\), first choose a tournament node with directed edges (large gray arrows) towards all tournament nodes that contain parts of \(S\).
            Within this tournament node, choose a set \(U\) from an \((n,k)\)-universal set with \(U \cap S = T\).
            There exists an extension vertex \(v\) with \(N(v) \cap S =~T\).
        }
        \label{fig:rado-construction}
    \end{figure}

    There are at most $2^{2^{3k \cdot 2}}$ tournaments  with $2^{3k}$ vertices.
    By enumerating them all,
    we find according to \Cref{lem:tournament}
    a tournament $F$ with the vertex set \([2^{3k}]\)
    such that, for every $S \subseteq [2^{3k}]$ of size $k$, there is a vertex $j \in [2^{3k}]$
    that has a directed edge towards every vertex in $S$.

    We will construct a graph \(G\) with vertex set \([n]\) satisfying the extension axioms \(\EA_k\).
    We start by using \Cref{thm:hashfunctions} to construct an $(n,k)$-universal set $\mathcal U$.
    Next, we need to partition the vertex set \([n]\) into $2^{3k}$ many parts $P_1,\dots,P_{2^{3k}}$ of size at least \(|\mathcal U|\) (these will correspond to vertices in a tournament).
    Let us argue that one can choose the parameter \(c\) such that this is possible for all $n \ge 2^{k^c}$.
    Let $c'$ be the constant from \Cref{thm:hashfunctions}, such that we can bound the size of our $(n,k)$-universal set by $2^{c'k} \log(n)$.
    We may choose \(c\) such, that for every $k \in \N$,
    \[
        2^{3k} \cdot 2^{c'k} \cdot 2 \le 2^{k^{c}}/k^c.
    \]
    Hence, for every $n \in \N$ with $2^{k^c} \le n$ holds \(2^{k^{c}}/k^c \le n/\log(n)\) and thus
    \[
        2^{3k} \cdot 2^{c'k} \log(n) \cdot 2 \le n.
    \]
    This means we can partition the vertices into parts \(P_1,\dots,P_{2^{3k}}\), each of size at least
    \[
        \lfloor 2^{c'k} \log(n) \cdot 2 \rfloor \ge |\mathcal U|.
    \]
    Next, we want to choose a function \(\pattern : [n] \to \mathcal U\).
    Since we have \(|P_j| \ge |\mathcal U|\), we can guarantee that \[\{ \pattern(v) \mid v \in P_j \} = \mathcal U\] for all \(j \in [2^{3k}]\).
    We further denote by \(\idx(v)\) the index \(j\) such that \(v \in P_j\).
    To construct our graph \(G\),
    we now do the following for all pairs \(\{u, v\} \in {[n] \choose 2}\):
    \begin{quote}
        Up to renaming, we may assume that the arc in our tournament \(F\) goes from \(\idx(v)\) to \(\idx(u)\).
        Add the edge \(\{u, v\}\) to \(G\) if and only if \(u \in \pattern(v)\).
    \end{quote}

    \paragraph*{Correctness}
    Let \(S \subseteq [n]\) and let \(T \subseteq S\).
    We constructed our tournament \(F\) using \Cref{lem:tournament} such that we can choose \(j \in [2^{3k}]\) with a directed edge to all vertices in \(\{\idx(s) \mid s \in S\}\).
    Since \(\mathcal U\) is an \((n,k)\)-universal set, there exists \(U \in \mathcal U\) such that \(S \cap U = T\).
    Since \(\pattern\), restricted to \(P_j\), was chosen to be surjective, we can pick \(v \in P_j\) with \(\pattern(v) = U\).
    Since \(F\) has no self loops, \(v \not\in S\).
    We selected \(v\) such that \(\idx(v)\) has in \(F\) a directed edge to \(\idx(u)\) for all \(u \in S\).
    Thus, \(v\) is adjacent to \(u \in S\) if and only if \(u \in \pattern(v) = T\).

    \paragraph*{Run time}
    The construction of $F$ takes time at most $2^{2^{3k \cdot 2}} \cdot k^2$.
    By \Cref{thm:hashfunctions}, we can construct the $(n,k)$-universal set in time $2^{c'k} n \log(n)$.
    The remaining part of the construction takes time $O(|\mathcal U| \cdot n^2)$.
    By possibly rescaling \(c\), we get a run time of \(O(2^{2^{ck}} n^2)\).
\end{proof}

\subsection{Relational structures}

The definition of \(k\)-extension axioms \(\EA^\sigma_k\) for general signatures \(\sigma\) becomes more involved,
and thus our finite Rado-construction becomes more opaque, even though we ultimately use the same arguments as for graphs.
Given a set \(S\), we find a tournament node oriented to the tournament nodes containing \(S\).
Instead of using universal sets, we now select a function \(f\) from a family of \((n,k)\)-perfect hash functions that is bijective on \(S\).
Our construction adds an extension element that is valid for all sets \(S\) on which \(f\) is bijective.
The definition of \(k\)-extension axioms \(\EA^\sigma_k\) (\Cref{def:EAstruc}) and the proof of \Cref{thm:radostructure} 
can be found in Appendix~\ref{apxPseudorandomStructures}.

\begin{restatable}{theorem}{RadoStructure}\label{thm:radostructure}
    Let \(\sigma\) be a relational signature.
    There exists a constant $c$ such that, for every $n$, one can construct
    in time \(O(n^{c})\) an \(n\)-element \(\sigma\)-structure satisfying the extension axioms
    $\EA^\sigma_{k}$ for all \(k \le \lfloor\log(\log(n))/c\rfloor\).
\end{restatable}

Even though we do not prove this, \Cref{thm:radostructure} can be easily modified to also generate pseudorandom hypergraphs.
As we see in the next section, just like for graphs, the extension axioms determine the truth of FO and LFP formulas and hold with high probability.

\subsection{Pseudorandom graphs and structures}\label{sec:three}

Let us now discuss why the constructions of \Cref{thm:radograph} and \Cref{thm:radostructure} give pseudorandom structures for FO and LFP.
First, observe that random graphs, hypergraphs, and relational structures readily satisfy the respective extension axioms.

\begin{restatable}{lemma}{EAProbs}\label{lem:EA-probs}
    For every signature \(\sigma\) and every \(k,t\),
    \begin{align*}
        \Pr_{G \sim G(1/2,n)}\bigl[G \not \models \EA_k \bigr] = {} & \negl(n), \\
        \Pr_{G \sim G_t(1/2,n)}\bigl[G \not \models \EA^t_k \bigr] = {} & \negl(n), \\
        \Pr_{\mathbb A \sim \struc{\sigma,n}}\bigl[ \mathbb A \not \models \EA^\sigma_k \bigr] = {} & \negl(n).
    \end{align*}
\end{restatable}
This is a well-known result \cite[Sections 4.1 and 4.2]{ParametricClasses}, but we include a proof in Appendix~\ref{apxEAprob} for completeness.

Next, we argue via pebbling games (see for example \cite[Lemma 12.7]{Libkin2004}) that the extension axioms determine the truth of FO or LFP sentences. 
This proof is also straightforward and included in Appendix \ref{appEA-determine-qtype}.

\begin{restatable}{lemma}{EAType}\label{lem:EA-determine-qtype}
    Consider a pair of graphs/hypergraphs/structures \(\mathbb A_1, \mathbb A_2\) of the same signature that both satisfy the corresponding \(k\)-extension axioms.
    Then, for every advice structure \(\mathbb X\), and every \textnormal{LFP} sentence \(\phi\)
    of quantifier rank at most \(k\) and matching signature, we have \( (\mathbb A_1,\mathbb X) \models \phi \Leftrightarrow (\mathbb A_2,\mathbb X) \models \phi\).
\end{restatable}

Now, \Cref{lem:EA-probs} and \Cref{lem:EA-determine-qtype} together with \Cref{thm:radograph} (for graphs) and \Cref{thm:radostructure} (for structures)
immediately imply the main theorem of this section.

\begin{theorem}\label{thm:RadoGenerator}
    There are deterministic polynomial-time algorithms to construct graphs or relational structures that are pseudorandom for \textnormal{FO} and \textnormal{LFP}.
\end{theorem}


\section{Pseudorandom Transductions}\label{secPRGClassification}

We now turn to the question of whether transductions can act as pseudorandom generators. Specifically, we ask for which logics $L_1$ and $L_2$ and for which relational signatures $\sigma$ and $\tau$ there exists an $(L_1, L_2)$-pseudorandom transduction from $\sigma$ to $\tau$, that is, an \(L_1\)-transduction that is secure against \(L_2\). In this section, we provide a complete classification for all pairs $L_1, L_2 \in \{\txt{FO}, \txt{LFP}, \txt{\LFPpar}\}$ and all relational signatures.

Clearly, the existence of such transductions depends on the number and arities of the relations in $\sigma$ and $\tau$. For instance, if the multiset of arities of $\tau$ is a subset of that of $\sigma$, we may simply match up the relation symbols with formulas of the form $$\theta_i(\seq{x}{a_i}) = R_j(\seq{x}{a_i})$$ and ignore the extra relation symbols from $\sigma$. This transduction is definable in any logic and will yield a uniformly random $\tau$-structure, so will be pseudorandom for any logic as well. A more interesting case is when the target $\tau$-structures inherently have more ``entropy'' than the input $\sigma$-structures, for example, going from a single binary relation to a single ternary relation. Our results for this case were summarized in Table \ref{tab23}.

In other cases it is not immediately clear how to measure entropy. For example, does there exist a pseudorandom transduction from a signature with one ternary relation to a signature with ten binary relations? Curiously, we find that the ``right'' way of measuring entropy depends on the tranduction logic $L_1$; in this particular case the answer for FO and LFP is No, while the answer for \LFPpar is Yes.

Let $\sigma = \langle \seq{R}{t} \rangle$ and $\tau = \langle R'_1, R'_2, \dots, R'_{t'} \rangle$ be signatures, where $R_i$ has arity $a_i$ and $R'_i$ has arity $a'_i$. We will show that the right measure of entropy for \LFPpar is the following ``lexicographic'' order \(\geqC\).
We say $\sigma \geqC \tau$ if and only if the arity-tuple $(\seq{a}{t})$ of \(\sigma\) 
is (weakly) lexicographically larger than the corresponding tuple $(a'_1, a'_2, \dots, a'_{t'})$ of \(\tau\) 
after sorting both tuples in descending order.
Equivalently,
\[
\sigma \geqC \tau \iff \sum_{i = 1}^{t} (t + t')^{a_i} \geq \sum_{i = 1}^{t'} (t + t')^{a'_i}.
\]
As one may expect, \(\geqC\) is a total order and matches our intuitive understanding of entropy in random structures, 
where one assigns an \(a\)-ary relation an entropy of \(n^a\), which is asymptotically larger than any combination of \(b\)-ary relations with \(b < a\).

The story becomes more complex when the transduction logic $L_1$ is FO or LFP.
Here, the correct entropy measure will be the order \(\geqS\) defined by
\[
    \sigma \geqS \tau \iff \txt{for all $k \in \zz_{\geq 1}$, } \sum_{i = 1}^{t} T(a_i, k) \geq \sum_{i = 1}^{t'} T(a'_i, k),
\]
where 
\[
    T(a, k) := \sum_{j = 0}^{k} (-1)^{k - j} j^a {k \choose j}
\]
counts the number of surjections from a set of size $a$ to a set of size $k$ \cite[(6.19)]{SurjectionFormula}.
Note that $\geqS$ is a refinement of $\geqC$, and is not a total order. 
That is, having more entropy in the sense of \(\geqC\) is a necessary, but not sufficient condition for the existence of pseudorandom generators.
Intuitively speaking, this is because FO and LFP are not powerful enough to convert certain ``shapes'' of the input entropy into the required output shape.

We summarize how the existence of pseudorandom transductions depends on both the logics and the entropies of the signatures in Table \ref{tabgeneral}, from which Table \ref{tab23} can be derived as a special case.

\begin{table}\begin{center}
    \begin{tabular}{|c|c|c|c|}
        \hline
        \backslashbox{$\Theta$}{$\phi$} & FO     & LFP    & \LFPpar                                       \\
        \hline
        FO                              & \xmarkg & \xmarkg & \xmarkg                                        \\
        \hline
        LFP                             & \xmarkg & \xmarkg & \xmarkg                                        \\
        \hline
        \LFPpar                         & \cmarkg & \cmarkg & {\color{blue}\thead{$\sigma \geqC \tau$ or\\($\sigma$ not all-unary\\and OWFs exist)}} \\
        \hline
    \end{tabular}
    \caption{Conditions on the signatures $\sigma$ and $\tau$ for the existence of an $(L_1, L_2)$-pseudorandom generator from $\sigma$ to $\tau$. The rows range over generator logics $L_1$ and the columns range over adversary logics $L_2$. This generalizes Table \ref{tab23}.}
    \label{tabgeneral}
\end{center}\end{table}

\subsection{Impossibility for \textup{LFP}}\label{subPRGImpossibility}

We now state and outline the proof of the following result, filling in the top two-thirds of Table \ref{tabgeneral}.

\begin{restatable}{theorem}{thmPRGImpossibility}\label{thmPRGImpossibility}
	Let $\sigma = \langle \seq{R}{t} \rangle$ and $\tau = \langle R'_1, R'_2, \dots, R'_{t'} \rangle$ be relational signatures, where $R_i$ has arity $a_i$ and $R'_i$ has arity $a'_i$. The following statements are equivalent:
	\begin{enumerate}[label={(\roman*)}]
		\item\label{itmPRGImpossibilityNormal} There exists an $(\textnormal{LFP}, \textnormal{FO})$-pseudorandom generator from $\sigma$ to $\tau$.
		\item\label{itmPRGImpossibilityStatistical} There exists a quantifier-free FO transduction $\Theta: \struc{\sigma} \to \struc{\tau}$ such that, for all $n$, the distribution on $\struc{\tau, n}$ obtained by applying $\Theta$ to $\mathbb{A} \sim \struc{\sigma, n}$ is statistically identical to the distribution $\mathbb{B} \sim \struc{\tau, n}$.
		\item\label{itmPRGImpossibilityFormula} $\sigma \geqS \tau$.
	\end{enumerate}
\end{restatable}

To see why this fills in the FO and LFP rows of Table \ref{tabgeneral}, first suppose $\sigma \geqS \tau$. Then by \ref{itmPRGImpossibilityStatistical}, there is a FO-transduction (and thus an LFP-transduction) $\Theta$ generating $\tau$-structures statistically identical to random $\tau$-structures. Thus, not even a computationally unbounded adversary could distinguish the two cases, so $\Theta$ is secure against any adversary logic, filling in the top two rows in the affirmative. On the other hand, if we do not have $\sigma \geqS \tau$, then by \ref{itmPRGImpossibilityNormal}, even a generator defined in LFP cannot fool every adversary in FO, and the same holds for weaker generators and stronger adversaries, filling in the top two rows in the negative.

This theorem can be thought of as an impossibility theorem for LFP. The equivalence of \ref{itmPRGImpossibilityNormal} and \ref{itmPRGImpossibilityStatistical} means that there are no ``novel'' $(\textnormal{LFP}, \textnormal{FO})$-pseudorandom generators, in the sense that any LFP-definable transduction that is secure even against the weak logic FO can be replaced by an extremely simple alternative transduction in which all formulas are FO and do not contain quantifiers.
And this simple transduction is ``as secure as it gets,'' in the sense that the output distribution is not merely computationally indistinguishable from the truly-random distribution, 
but \emph{is} the truly-random distribution.
Thus, all the additional power of LFP is irrelevant for generating pseudorandom structures in this context.

As we previously remarked, the notion of entropy $\geqS$ for characterizing pseudorandom generators in these logics is only a partial order. For instance, if $\sigma$ consists of one binary relation and $\tau$ consists of two unary relations, Theorem \ref{thmPRGImpossibility} implies that there are no pseudorandom generators either from $\sigma$ to $\tau$ or from $\tau$ to $\sigma$, since
$$\sum_{i = 1}^t T(a_i, 2) = 2 > 0 = \sum_{i = 1}^{t'} T(a'_i, 2)$$
but
$$\sum_{i = 1}^t T(a_i, 1) = 1 < 2 = \sum_{i = 1}^{t'} T(a'_i, 1).$$

The full proof is long and technical, and thus deferred to Appendix \ref{appPRGImpossibilityProof}. Here we sketch the argument for all 3 cases.

\ref{itmPRGImpossibilityStatistical}$\implies$\ref{itmPRGImpossibilityNormal}: This is immediate, since a transduction $\Theta$ with the properties in \ref{itmPRGImpossibilityStatistical} is an $(L_1, L_2)$-pseudorandom generator for any $L_1$ and $L_2$ such that $L_1$ extends quantifier-free FO.

\ref{itmPRGImpossibilityFormula}$\implies$\ref{itmPRGImpossibilityStatistical}:
Consider the following example. Suppose 
we want to build a transduction from an input signature \(\sigma\) consisting
of three relations of arities $a_1 = 3$, $a_2 = 1$, $a_3 = 1$,
to an output signature \(\tau\) with also three relations, but of arities
$a'_1 = a'_2 = a'_3 = 2$. 
If $\sigma$ had a binary relation, we could use it to directly fill one of the three relations in $\tau$; but since we do not, we must use specializations of the ternary relation of the form $R_1(x, x, y)$, $R_1(x, y, x)$, and $R_1(y, x, x)$. For $x \neq y$, these tuples will each be in the relation $R_1$ independently of one another, so can each be used to fill one of the 3 binary relations in $\tau$. However, this accounting only takes care of ``binary parts'' of the relations, where there are two distinct variables involved. To determine which relations $R'_i(x, x)$ hold, we may only use $R_1(x, x, x)$ once, so have to use the unary relations $R_2(x)$ and $R_3(x)$ for the other two. Thus, by treating the binary and unary parts separately, we arrive at a quantifier-free FO transduction $\Theta = (\theta_1, \theta_2, \theta_3)$ that produces uniformly random $\tau$-structures from uniformly random $\sigma$-structures:
\begin{align*}
    \theta_1(x_1, x_2) :=\ & R_1(x_1, x_1, x_2)\\
    \theta_2(x_1, x_2) :=\ & \bigl((x_1 \neq x_2) \wedge R_1(x_1, x_2, x_1)\bigr) \vee \bigl((x_1 = x_2) \wedge R_2(x_1)\bigr)\\
    \theta_3(x_1, x_2) :=\ & \bigl((x_1 \neq x_2) \wedge R_1(x_1, x_2, x_2)\bigr) \vee \bigl((x_1 = x_2) \wedge R_3(x_1)\bigr).
\end{align*}
More generally, the condition that $\sigma \geqS \tau$ amounts to checking that there are enough $k$-ary parts of relations in $\sigma$ to cover the $k$-ary parts of the relations in $\tau$, for each possible number of distinct variables $k$.

$\neg$\ref{itmPRGImpossibilityFormula}$\implies$$\neg$\ref{itmPRGImpossibilityNormal}: This is the most difficult part of the argument, where we must use the zero-one law to argue that, no matter what the LFP transduction $\Theta = (\seq{\theta}{t'})$ does, we can construct an FO sentence $\phi$ that distinguishes its outputs from truly random $\tau$-structures. In order to do this, we must manipulate each formula $\theta_i$ to put it into a form where the zero-one law applies. For example, suppose one of the formulas is
$$\theta_i(x_1, x_2, x_3) = \forall x_4\ R_1(x_3, x_4).$$
By distinguishing whether or not \(x_4 \in \{ x_1,x_2,x_3\}\), we first equivalently rewrite this sentence as
\begin{equation*}
    \theta_i(x_1, x_2, x_3) = R_1(x_3, x_1) \wedge R_1(x_3, x_2) \wedge R_1(x_3, x_3) \wedge \forall x_4 \notin \{x_1, x_2, x_3\}\ R_1(x_3, x_4).
\end{equation*}
We then observe that $R_1(x_3, \cdot)$ can be thought of as a random unary relation on the structure obtained by removing $x_1$, $x_2$, and $x_3$, so by the zero-one law applied to this structure, we know that the final part of this sentence,
$$\forall x_4 \notin \{x_1, x_2, x_3\}\ R_1(x_3, x_4),$$
can be replaced by either TRUE or FALSE (in this case FALSE) and the resulting formula will be equivalent with probability $1 - \negl(n)$. In this manner, we may remove all quantifiers from the formula. Taking a union bound over the polynomially-many negligible error terms, we conclude that the resulting quantifier-free formula $\overline{\theta}_i$ will be equivalent with probability $1 - \negl(n)$.

Finally, we let $k$ be a positive integer violating the condition that $\sigma \geqS \tau$ and choose a large (but constant) positive integer $c$. We write a sentence $\phi$ that checks the following condition: ``For every way that a $c$-tuple $(\seq{y}{c})$ of $c$ distinct elements of the universe \emph{could} behave with respect to relations using at most $k$ distinct variables exclusively from $\{\seq{y}{c}\}$, there is in fact a $c$-tuple that \emph{does} behave in that way.'' Clearly, a truly random $\tau$-structure will satisfy $\phi$ with probability $1 - \negl(n)$. On the other hand, we show that, for large enough $c$, $\overline{\Theta} = (\seq{\overline{\theta}}{t'})$ produces structures satisfying $\phi$ with probability 0. This step is formally accomplished by the following Lemma, which is proved in Appendix \ref{appKTypeProof} and is also used later to show impossibilities for \LFPpar when $\sigma$ has only unary relations.

\begin{definition}\label{defCKType}
    Given integers $1 \leq k \leq c$, a relational signature $\sigma = \langle \seq{R}{t} \rangle$ with arities $\seq{a}{t}$, a $\sigma$-structure $\mathbb{A}$ and distinct elements $\seq{y}{c} \in \mathbb{A}$, the \emph{$(c, k)$-type of $(\seq{y}{c})$ in $\mathbb{A}$} is the subset
    \begin{equation*}\label{equCKTypeSet}
        T\subseteq \bigcup_{i = 1}^t \bigcup_{\substack{S \subseteq [c]\\\abs{S} \leq k}} \{(i, S, g) \suchthat g: [a_i] \to S \txt{ is surjective}\}
    \end{equation*}
    such that \(T\) describes (by indexing $(\seq{y}{c})$) exactly which relations in $\mathbb{A}$ hold on tuples containing at most $k$ distinct elements from $(\seq{y}{c})$. In other words, $(i, S, g) \in T$ means $(y_{g(1)}, y_{g(2)}, \dots, y_{g(a_i)}) \in R_i^\mathbb{A}$.
\end{definition}

\begin{restatable}{lemma}{lemKType}\label{lemKType}
    Suppose $\sigma$ and $\tau$ are relational signatures for which some positive integer $k$ violates the condition that ${\sigma \geqS \tau}$. Let $f: \struc{\sigma} \to \struc{\tau}$ be a function such that, for any positive integer $c \geq k$ and any $\sigma$-structure $\mathbb{A}$, if the $(c, k)$-types of any pair of $c$-tuples $(\seq{y}{c}), (y'_1, y'_2, \dots, y'_c) \in \mathbb{A}^c$ are the same in $\mathbb{A}$ then their $(c, k)$-types are the same in $f(\mathbb{A})$. Then the distribution obtained by applying $f$ to a uniformly random $\sigma$-structure is not pseudorandom for \textnormal{FO}.
\end{restatable}

We may apply this lemma to $f := \overline{\Theta}$ since it is quantifier-free, which yields that $\overline{\Theta}$ is not a pseudorandom generator. Since the original transduction $\Theta$ differs from $\overline{\Theta}$ only negligibly, it follows that $\Theta$ is not a pseudorandom generator either.

\subsection{Possibility for \LFPpar}\label{subPRGPossibility}

Adding in a parity operator allows us to circumvent the zero-one law so that pseudorandom transductions are plausible. In the presence of at least one non-unary relation, we can canonize an input structure and thus define a pseudorandom output structure for FO and LFP using the polynomial-time construction from \Cref{thm:RadoGenerator}.

\begin{corollary}\label{corFPCToLFPPRG}
    For $L_2 \in \{\textnormal{FO}, \textnormal{LFP}\}$, there exists an $(\textnormal{\LFPpar}, L_2)$-pseudorandom generator from $\sigma$ to $\tau$ if and only if $\sigma$ has at least one non-unary relation or $\sigma$ and $\tau$ both have only unary relations, with $\sigma$ containing at least as many as \(\tau\).
\end{corollary}

\begin{proof}
    If $\sigma$ has at least one non-unary relation, it is possible to canonize (meaning define a total order on) the entire input structure with probability $1 - \negl(n)$ by Lemma \ref{lemRandomCanonization}. Assuming this happens, by the Immerman-Vardi Theorem (Theorem \ref{thmImmermanVardi}) we have the full power of polynomial-time computation, so may apply 
    the deterministic algorithm of \Cref{thm:RadoGenerator}.
    
    If $\sigma$ and $\tau$ both have only unary relations and $\sigma$ has at least as many as \(\tau\), a pseudorandom transduction obviously exists by simply matching them up and ignoring any extras.

    If neither of these conditions hold, it means that $\sigma \geqS \tau$ is violated for $k = 1$. Let $\Theta: \struc{\sigma} \to \struc{\tau}$ be an arbitrary \LFPpar transduction. For any integer $c$, if two $c$-tuples have the $(c, 1)$-types in a $\sigma$-structure $\mathbb{A}$, it means that the same unary relations hold on the $i\tth$ element of each tuple for each $1 \leq i \leq c$. Since $\sigma$ has only unary relations, it follows that there is an automorphism of $\mathbb{A}$ taking one tuple to the other. Since $\Theta$ is a transduction, it means that there is a similar automorphism in $\Theta(\mathbb{A})$, implying that the two tuples have the same $(c, 1)$-type in $\Theta(\mathbb{A})$. Hence Lemma \ref{lemKType} applies, so $\Theta$ is not a pseudorandom transduction for FO or LFP.
\end{proof}

In light of this theorem, the only remaining box in Table~\ref{tabgeneral} we have not discussed is the bottom-right cell, when both \mbox{logics} are \LFPpar.
The following result completely characterizes when $(\textnormal{\LFPpar},
\textnormal{\LFPpar})$-pseudorandom generators exist, depending on the comparative
entropies of $\sigma$ and $\tau$ (this time measured using $\geqC$, rather than
$\geqS$).

\begin{restatable}{theorem}{thmFPCPRG}\label{thmFPCPRG}
    Let $\sigma$ and $\tau$ be relational signatures.
    \begin{enumerate}[label={(\roman*)}]
        \item\label{itmFPCPRGConditionallyExists} If $\sigma \lC \tau$ and $\sigma$ contains at least one non-unary relation, then $(\textnormal{\LFPpar}, \textnormal{\LFPpar})$-pseudorandom generators from $\sigma$ to $\tau$ exist if and only if length-increasing ordinary pseudorandom generators exist (i.e., Conjecture \ref{cnjOWF} holds).
        \item\label{itmFPCPRGExists} If $\sigma \geqC \tau$, then there exists an $(\textnormal{\LFPpar}, \textnormal{\LFPpar})$-pseudorandom generator from $\sigma$ to $\tau$.
        \item\label{itmFPCPRGDoesNotExist} Otherwise, there does not exist an $(\textnormal{\LFPpar}, \textnormal{\LFPpar})$-pseudorandom generator from $\sigma$ to $\tau$.
    \end{enumerate}
\end{restatable}

The most difficult and interesting statement is \ref{itmFPCPRGConditionallyExists}, and the proof involves converting back and forth between strings and relational structures in a way that is definable in \LFPpar. It turns out that the backward direction of the proof, which involves building a pseudorandom transduction from an ordinary pseudorandom generator, is relatively straightforward. However, the forward direction presents multiple challenges. To build an ordinary pseudorandom generator out of a pseudorandom transduction $\Theta$, the basic idea is to use the random input bits to construct a structure $\mathbb{A}$, simulate applying $\Theta$ to it, then read out the relations to generate more random bits than we started with. To prove security, we must argue that an ordinary adversary distinguishing random from pseudorandom strings can be turned into a logical adversary distinguishing random from pseudorandom structures. To define such a logical adversary, it is necessary for the post-processing step of reading out the random bits from the relations to be definable in \LFPpar in an order-invariant way, i.e., \emph{without} referring to the original arbitrary order of the elements of the universe. When $\tau$ contains an extra $k$-ary relation for $k \geq 2$ we may achieve this by using Lemma \ref{lemRandomCanonization} to do the post-processing step. Thus, the only remaining case is when $\tau$ contains the same number of $k$-ary relations as $\sigma$ for all $k \geq 2$ but has more unary relations. This problem is by far the most challenging, and we tackle it by employing a novel adaptation of the length-extension theorem from cryptography (reviewed in Section \ref{subCryptoBackground}) to our domain. This involves writing an \LFPpar formula to iterate $\Theta$ a polynomial number of times, extracting a single random parity bit from the extra unary relation in $\tau$ on each iteration, as illustrated in Figure \ref{figLengthExtension}.

The only place where we use the parity operator is in applying Lemma \ref{lemRandomCanonization} (multiple times). The proof goes through exactly the same if we were to replace \LFPpar with the extensively-studied \emph{fixed point logic with counting} (FPC).

\begin{proof}[Proof of Theorem \ref{thmFPCPRG}]
    \ref{itmFPCPRGConditionallyExists}: We define
    \begin{align*}
        p(n) &:= \sum_{i = 1}^{t} n^{a_{i}}, & p'(n) &:= \sum_{i = 1}^{t'} n^{a'_{i}}.
    \end{align*}
    For the backward direction, suppose ordinary pseudorandom generators exist. Let $R$ be a relation in $\sigma$ of arity at least 2 and let $f: \{0, 1\}^n \to \{0, 1\}^{p'(n)}$ be a pseudorandom generator. Consider the \LFPpar transduction $\Theta$ described by the following algorithm. On input $\mathbb{A}$, we first apply the transduction from Lemma \ref{lemRandomCanonization} using $R$ to obtain (with probability $1 - \negl(n)$) a total order over the universe $\leq$ and at least $\Omega(n^2)$ bits of independent randomness from $R$ at locations described by relation $Q$. We then simulate $f$ (via Theorem \ref{thmImmermanVardi}) on the first $n$ bits of randomness according to $\leq$ to obtain $p'(n)$ bits of pseudorandomness, which we then use to fill all the relations in $\tau$, again according to the order $\leq$. It is clear that this can all be accomplished as an \LFPpar transduction.
    
    Suppose toward a contradiction that some \LFPpar adversary $\phi$ could break the security of $\Theta$. Then consider the algorithm $A$ that, on input $t \in \{0, 1\}^{p'(n)}$, generates a $\tau$-structure $\mathbb{B}$ using $t$ to determine each relation, then evaluates $\phi$ and outputs 1 if and only if $\mathbb{B} \models \phi$. Clearly, when $t$ is a uniformly random string, $A$ outputs 1 with the same probability that $\mathbb{B} \models \phi$ for a random $\mathbb{B} \sim \struc{\tau, n}$. 
    On the other hand, feeding in $t = f(s)$ for a uniformly random $s \sim \{0, 1\}^n$ is equivalent to attempting to apply Lemma \ref{lemRandomCanonization} over and over until it succeeds in ordering the structure and producing $n$ ordered bits of independent randomness, then applying $f$ to the result, as depicted in Figure \ref{figBackwardDirection}. On a random structure $\mathbb{A}$ that can be successfully ordered, $A$ will then output 1 if and only if $\Theta(\mathbb{A}) \models \phi$. Thus, the probability $A$ outputs 1 in the case where $t = f(s)$ is precisely the probability that $\Theta(\mathbb{A}) \models \phi$ for a random $\mathbb{A} \sim \struc{\sigma, n}$, but conditioned on the event that $\mathbb{A}$ was successfully ordered, which we denote as sampling $\mathbb{A} \sim \ostruc{\sigma, n}$. By Lemma \ref{lemRandomCanonization}, this is only different from the unconditional $\Pr_{\mathbb{A} \sim \struc{\sigma, n}}[\Theta(\mathbb{A}) \models \phi]$ by a negligible amount. By the triangle inequality, we have that
    \begin{align*}
        \abs{\Pr_{t \sim \{0, 1\}^{p'(n)}}[A(t) = 1] - \Pr_{s \sim \{0, 1\}^{n}}[A(f(s)) = 1]} &= \abs{\Pr_{\mathbb{B} \sim \struc{\tau, n}}[\mathbb{B} \models \phi] - \Pr_{\mathbb{A} \sim \ostruc{\sigma, n}}[\Theta(\mathbb{A}) \models \phi]}\\
        &\geq \abs{\Pr_{\mathbb{B} \sim \struc{\tau, n}}[\mathbb{B} \models \phi] - \Pr_{\mathbb{A} \sim \ostruc{\sigma, n}}[\Theta(\mathbb{A}) \models \phi]}\\
        &\ \ \ - \abs{\Pr_{\mathbb{A} \sim \ostruc{\sigma, n}}[\Theta(\mathbb{A}) \models \phi] - \Pr_{\mathbb{A} \sim \struc{\sigma, n}}[\Theta(\mathbb{A}) \models \phi]},
    \end{align*}
    which is non-negligible, since the first term is non-negligible and the second term is negligible. Thus, $A$ breaks the security of $f$, which is a contradiction.

    \begin{figure}
        \begin{center}
            \includegraphics[scale=0.88]{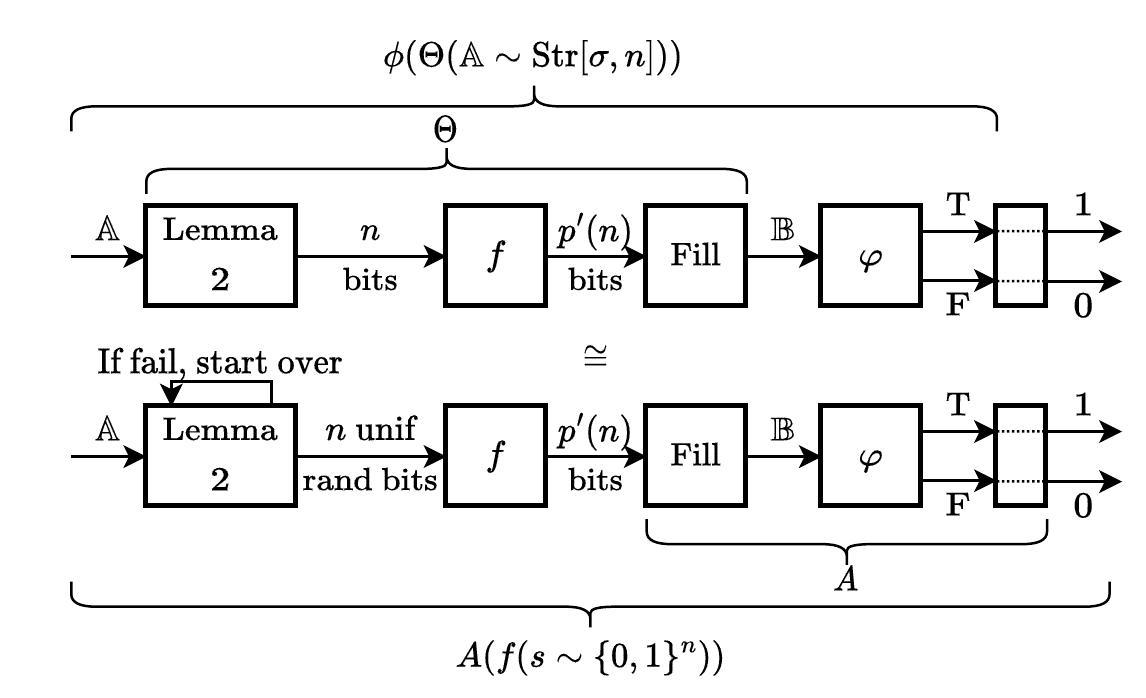}
        \end{center}
        \caption{Illustration accompanying the backward direction of the proof of Theorem \ref{thmFPCPRG} \ref{itmFPCPRGConditionallyExists}, where $\cong$ denotes the fact that the two distributions over outputs in $\{0, 1\}$ are only negligibly different, as Lemma \ref{lemRandomCanonization} fails with negligible probability.}
        \label{figBackwardDirection}
    \end{figure}
    
    For the forward direction, suppose $\Theta: \struc{\sigma} \to \struc{\tau}$ is an $(\textnormal{\LFPpar}, \textnormal{\LFPpar})$-pseudorandom generator. To define an ordinary pseudorandom generator, there are two cases to consider. Let $k$ be the largest arity such that $\sigma$ and $\tau$ have a different number of relations of arity $k$, and first consider the case where $k \geq 2$. Let $g(n) = \Omega(n^k)$ be the function from \Cref{lemRandomCanonization}, and let
    $$h(n) := \sum_{k' = k}^\infty \abs{\{i \suchthat a_i = k'\}} \cdot n^{k'}.$$
    From the way $k$ was defined and the fact that $\sigma \lC \tau$, we know that, for all $k' \geq k$, the coefficient of $n^{k'}$ is the same in $h(n)$ and $p'(n)$, while the $n^k$ coefficient is strictly larger in $h(n)$. Hence, for large enough $n$, $p'(n) \leq h(n) - n^k + g(n)$, as $g(n) = \Omega(n^k)$ dominates any remaining degree $< k$ terms coming from relations in $\tau$ of arity $< k$. It thus suffices to define a pseudorandom generator $f: \{0, 1\}^{p(n)} \to \{0, 1\}^{h(n) - n^k + g(n)}$. On input $s \in \{0, 1\}^{p(n)}$, we use all of the random bits to generate a random $\sigma$-structure $\mathbb{A}$. We then run $\Theta$ to obtain a $\tau$-structure $\mathbb{B} = \Theta(\mathbb{A})$, on which we apply Lemma \ref{lemRandomCanonization} using an arbitrary $k$-ary relation $R$ of $\tau$. With probability $1 - \negl(n)$, we obtain a total order over the universe and can extract and output $g(n)$ ordered bits of randomness from $R$ and a further $h(n) - n^k$ ordered bits of randomness from the other relations of arity at least $k$. (If the ordering step fails, we just output some arbitrary string.)
    
    Suppose toward a contradiction that some adversary $A$ could break the security of $f$. Then we describe an \LFPpar adversary $\phi$ to break $\Theta$ via the following algorithm. Given a $\tau$-structure $\mathbb{B}$, we order the structure and extract $h(n) - n^k + g(n)$ bits of randomness as in the last step of $f$, and then return true if and only if $A$ outputs 1 on the resulting string. Clearly, this can be implemented in \LFPpar (by Theorem \ref{thmImmermanVardi}). If $\mathbb{B}$ was sampled randomly from $\struc{\tau, n}$, then, with probability $1 - \negl(n)$, the input fed to $A$ is a uniformly random $p(n)$-bit string. On the other hand, testing whether $\mathbb{B} = \Theta(\mathbb{A}) \models \phi$ for $\mathbb{A} \sim \struc{\sigma, n}$ is equivalent to generating a random $p(n)$-bit string, using it to fill relations in $\mathbb{A}$, applying $\Theta$, applying Lemma \ref{lemRandomCanonization} to extract $h(n) - n^k + g(n)$ bits of randomness, and finally feeding the resulting string into $A$. Executing these steps is precisely equivalent to running $f$ on a uniformly random $p(n)$ bit string and then feeding the result to $A$. Thus, since $A$ breaks the security of $f$, it follows from a similar argument as in the forward direction that $\phi$ breaks the security of $\Theta$, contradicting our assumption.
    
    Finally, consider the case where $k = 1$. In other words, $\sigma$ and $\tau$ have the same numbers of relations of every arity greater than 1, but $\tau$ has strictly more unary relations than~$\sigma$. Without loss of generality, assume $\sigma$ and $\tau$ have the same relation symbols, except that $\tau$ has exactly one more unary relation symbol, $U$ (if there are more than one extra unary relations, we ignore them). Let $f: \{0, 1\}^{p(n)} \to \{0, 1\}^{p(n) + 1}$ be defined as follows. We use the input bits to generate a random $\sigma$-structure $\mathbb{B}_0$, then apply $\Theta$ repeatedly, $p(n) + 1$ times, ignoring the $U$ relation. Let $\mathbb{B}_i$ be the intermediate $\tau$-structure obtained after $i$ iterations. We then output a string $t$ where the $i\tth$ bit of $t$ is equal to $\abs{U^{\mathbb{B}_i}} \bmod 2$, i.e., the parity of the number of elements on which $U$ held on the output of the $i\tth$ iteration (see Figure \ref{figLengthExtension}).

    \begin{figure}
        \begin{center}
            \includegraphics[scale=0.88]{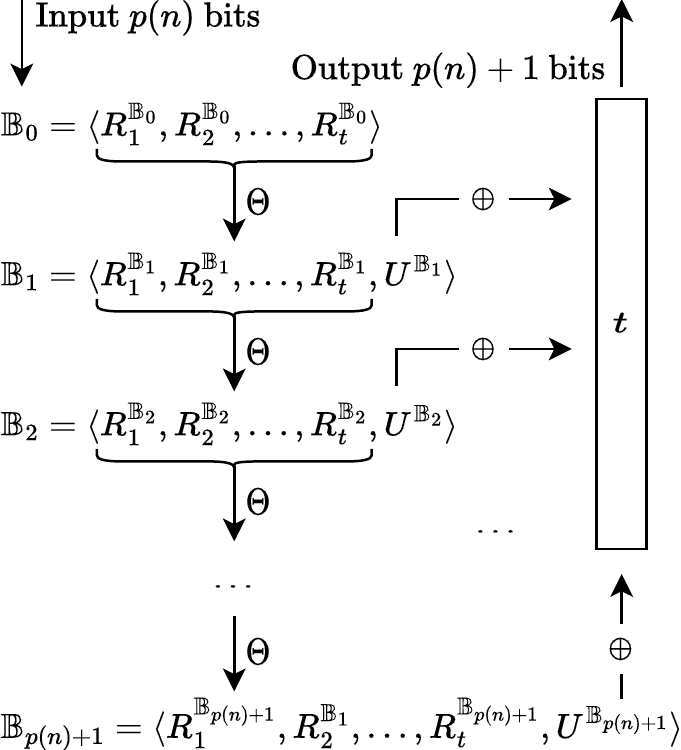}
        \end{center}
        \caption{Illustration of the length-extension construction of an ordinary pseudorandom generator $f: \{0, 1\}^{p(n)} \to \{0, 1\}^{p(n) + 1}$ using a pseudorandom transduction $\Theta$.}
        \label{figLengthExtension}
    \end{figure}
    
    To prove $f$ is secure, we use the hybrid argument. For each $i \in \{0, 1, 2, 3, \dots, p(n) + 1\}$, let $H_i$ be the distribution over ($p(n) + 1$)-bit strings where the first $i$ bits are chosen uniformly at random, while the remaining bits are determined by iterating $\Theta$ as in the definition of $f$, starting from a random $\sigma$-structure. Suppose toward a contradiction that some adversary $A$ could break the security of $f$. That is,
    $$\abs{\Pr_{s \sim \{0, 1\}^{p(n)}}[A(f(s)) = 1] - \Pr_{t \sim \{0, 1\}^{p(n) + 1}}[A(t) = 1]}$$
    is non-negligible. Since $H_0$ is equivalent to the former distribution and $H_{p(n) + 1}$ is equivalent to the latter distribution, it follows from the triangle inequality that, for some $i^* \in \{0, 1, 2, 3, \dots, p(n)\}$,
    $$\abs{\Pr_{y \sim H_{i^*}}[A(y) = 1] - \Pr_{y \sim H_{i^* + 1}}[A(y) = 1]}$$
    is non-negligible (as the sum of any $p(n)$ negligible functions is negligible). We will use this to define an \LFPpar adversary $\phi$ to break $\Theta$. First consider the following randomized algorithm $A'$. Given a $\tau$-structure $\mathbb{B}$, we uniformly sample $r \sim \{0, 1\}^{i^*}$, then iterate $\Theta$, $p(n) - i^*$ times, starting from $\mathbb{B}$, to extract $p(n) + 1 - i^*$ parity bits from the $U$ relation, just as in the definition of $f$ (we include the initial parity of $U^\mathbb{B}$ as well), appending the result to $r$. We then run $A$ on the resulting $p(n) + 1$ bit string and output whatever it does. Observe that, when the input $\mathbb{B}$ is a random $\tau$-structure, the distribution over the input fed to $A$ is precisely $H_{i^* + 1}$, and when $\mathbb{B} = \Theta(\mathbb{A})$ for $\mathbb{A} \sim \struc{\sigma, n}$, the distribution over the input fed to $A$ is precisely $H_{i^*}$. Thus, we have that
    $$\abs{\Pr_{\substack{\mathbb{A} \sim \struc{\sigma, n}\\r \sim \{0, 1\}^{i^*}}}[A'(\Theta(\mathbb{A})) = 1] - \Pr_{\substack{\mathbb{B} \sim \struc{\tau, n}\\r \sim \{0, 1\}^{i^*}}}[A'(\mathbb{B}) = 1]}$$
    is non-negligible. All that remains is to convert $A'$ into an \LFPpar sentence $\phi$, which will then contradict our assumption that $\Theta$ is secure. The main difficulty is that $A'$ is a randomized algorithm. While an \LFPpar-transduction cannot draw random bits for $r$ we can instead encode a good realization $r^*$ into the advice structure and use that. By this, we mean a sample $r \in \{0, 1\}^{i^*}$ such that
    \begin{align*}
        \abs{\Pr_{\mathbb{A} \sim \struc{\sigma, n}}[A'(\Theta(\mathbb{A}), r^*) = 1] - \Pr_{\mathbb{B} \sim \struc{\tau, n}}[A'(\mathbb{B}, r^*) = 1]}\\
        \geq \abs{\Pr_{\substack{\mathbb{A} \sim \struc{\sigma, n}\\r \sim \{0, 1\}^{i^*}}}[A'(\Theta(\mathbb{A})) = 1] - \Pr_{\substack{\mathbb{B} \sim \struc{\tau, n}\\r \sim \{0, 1\}^{i^*}}}[A'(\mathbb{B}) = 1]},
    \end{align*}
    which must exist by the averaging principle. We also encode a total order into the advice structure to allow us to store the outputs of each iteration of $\Theta$ using an inductively defined relation. Thus, we can turn $A'$ into an \LFPpar sentence $\phi$, concluding the proof of case \ref{itmFPCPRGConditionallyExists}.

    \ref{itmFPCPRGExists}: If $(a'_1, a'_2, \dots, a'_{t'})$ is a permutation of $(\seq{a}{t})$, then clearly an $(\txt{\LFPpar}, \txt{\LFPpar})$-pseudorandom generator exists by mapping relations to each other according to that permutation. Otherwise, let $k$ be the largest arity such that $\sigma$ and $\tau$ have a different number of relations of arity $k$. If $k = 1$, then it is again easy to construct a pseudorandom generator from $\sigma$ to $\tau$: Just match up relations by arity and ignore the extra unary relation(s) in $\sigma$. If $k \geq 2$, then matching up the lower arity parts may not be possible (that is, when $\sigma \geqC \tau$ but $\sigma \not\geqS \tau$). Instead, we appeal to Lemma \ref{lemRandomCanonization} again, using a $k$-ary relation to (with probability $1 - \negl(n)$) order the universe and then fill the lower arity relations with the $\Omega(n^k)$ bits of randomness defined by $Q$, in order. Since this yields a uniformly random $\tau$-structure with probability $1 - \negl(n)$, the resulting distribution is clearly pseudorandom for any logic.
    
    \ref{itmFPCPRGDoesNotExist}: In this case $\sigma$ contains only unary relations and $\sigma \lC \tau$, which implies $\sigma \not \geqS \tau$ as $\geqS$ is a refinement of $\geqC$. By Corollary \ref{corFPCToLFPPRG}, there does not even exist an $(\txt{\LFPpar}, \txt{FO})$-pseudorandom generator from $\sigma$ to $\tau$, so clearly there cannot be an $(\txt{\LFPpar}, \txt{\LFPpar})$-pseudorandom generator either.
\end{proof}

\section{First Order Logic with Parity}\label{secFOPlusParity}

The transduction \(\theta^t(x_1,\dots,x_{t}) := \oplus y~ \bigwedge_{i=1}^t E(y,x_i)\)
asks whether there are an odd number of vertices \(y\) that are adjacent to all vertices from \(x_1,\dots,x_t\).
It creates a \(t\)-hypergraph from an ordinary graph.
We argue that FO and LFP cannot distinguish \(\theta^t(G(n,1/2))\) from a real random \(t\)-hypergraph \(G_t(n,1/2)\)
by showing that both satisfy the following extension axioms with high probability.

\begin{definition}\label{def:EA-hypergraph}
    A \(t\)-hypergraph satisfies the extension axioms \(\EA^t_k\) if for
    all sets \(S\) of size \(k\)
    and all \(T \subseteq {S \choose t-1}\)
    there exists \(v \not \in S\) such that
    for every \( \{ s_1,\dots, s_{t-1} \} \in {S \choose t-1}\)
    we have a hyperedge \(\{s_1,\dots,s_{t-1},v\}\)
    if and only if \(\{ s_1,\dots, s_{t-1} \} \in T\).
\end{definition}

\Cref{lem:EA-determine-qtype} states that, also for hypergraphs, the corresponding extension axioms determine the truth of LFP sentences, even in the presence of an advice structure,
and it follows from \Cref{lem:EA-probs} that a random hypergraph satisfies the extension axioms with high probability:
For every \(k\) and \(t\),
\[
    \Pr_{G \sim G_t(1/2,n)}\bigl[ G \not \models \EA^t_k \bigr] = \negl(n).
\]

As the main result of this section, a probabilistic analysis of parities of certain sets in \(G(n,1/2)\) will
reveal that the same also holds in the transduced graph \(\theta^t(G(n,1/2))\).

\begin{restatable}{lemma}{TransductionSatisfiesEM}\label{lem:transductionSatisfiesEM}
    For every \(k\) and \(t\),
    \[
        \Pr_{G \sim G(n,1/2)}\bigl[ \theta^t(G) \not \models \EA^t_k \bigr] = \negl(n).
    \]
\end{restatable}

We prove this lemma in the next subsection.
Then the following theorem immediately follows
from Lemmas \ref{lem:EA-probs}, \ref{lem:EA-determine-qtype}, and \ref{lem:transductionSatisfiesEM}.

\begin{theorem}\label{thm:ParityGenerator}
    For every \(t\), \(\theta^t\) is an $(\textnormal{\FOpar},\textnormal{\textnormal{LFP}})$-pseudorandom generator from graphs to \(t\)-hypergraphs.
\end{theorem}

It would be interesting to find a logic that admits pseudorandom generators that are secure against adversaries in the same logic.
We believe \FOpar to be a candidate for this.
The current proof only shows security against FO or LFP adversaries, since the extension axioms are not enough to determine the truth value of \FOpar sentences.
However, Kolaitis and Kopparty~\cite{FOPlusParity} give a powerful result for \FOpar in a similar spirit:
To determine the truth value of any \FOpar sentence, it is sufficient to know
certain ``subgraph frequencies'', that is, the parities of the number of occurrences of certain subgraphs up to some size.
If this statement can be lifted to hypergraphs, then, to fool \FOpar, it is sufficient to transduce a \(t\)-hypergraph that mimics these subgraph frequencies.
Thus, \cite{FOPlusParity} can be understood as evidence for the following conjecture.

\begin{conjecture}
For every \(t\), \(\theta^t\) is an $(\textnormal{\FOpar},\textnormal{\FOpar})$-pseudorandom generator from graphs to \(t\)-hypergraphs.
\end{conjecture}

\subsection{Parity Probabilities in Random Graphs}\label{sec:parityRandomGraphs}

It remains to show that \(\theta^t(G(n,1/2))\) satisfies the extension axioms with high probability.
Roughly speaking, the sets \(\{v\} \cup C\) in the following lemma with \(|C|=t-1\) will correspond to the input variables of \(\theta^t\), that is,
the parity of the ``number of vertices that are adjacent to every vertex in \(\{v\} \cup C\)''
will determine whether \(\{v\} \cup C\) forms a hyperedge.

\begin{restatable}{lemma}{ParityLemma}\label{lem:parity2}
    Let \(k \in \N\). In \(G(n,1/2)\), with probability \(1-\negl(n)\),
    for every vertex set \(S \subseteq [n]\) of size \(k\) and every \(\mathcal{C} \subseteq \{ C \mid \emptyset \neq C \subseteq S \}\) there exists \(v \not \in S\) such that
    for all \(\emptyset \neq C \subseteq S\) we have \(C \in \mathcal{C}\) if and only if
    the number of vertices that are adjacent to every vertex in \(\{v\} \cup C\) is odd.
\end{restatable}

We derive \Cref{lem:parity2} via a careful basis construction over a suitable vector space from the following similar statement, which we prove first.

\begin{lemma}\label{lem:parity1}
    Let \(k \in \N\). In \(G(n,1/2)\), with probability \(1-\negl(n)\),
    for every vertex set \(S \subseteq [n]\) of size \(k\) and every \(\mathcal{B} \subseteq \{ B \mid \emptyset \neq B \subseteq S \}\), there exists \(v \not \in S\) such that
    for all \(\emptyset \neq B \subseteq S\) we have \(B \in \mathcal{B}\) if and only if
    the number of vertices that are adjacent to every vertex in \(\{v\} \cup B\) and non-adjacent to every vertex in \(S \setminus B\) is odd.
\end{lemma}
\begin{proof}
    The distribution of \(G(n,1/2)\) is identical to a random process where all adjacencies between tuples \(\{u, v\} \in {[n] \choose 2}\) are revealed one by one in an arbitrary order,
    and each tuple is revealed to be either adjacent or non-adjacent with probability \(1/2\), each.

    Let \(S\) be a set of \(k\) vertices and \(\mathcal{B} \subseteq \{ B \mid \emptyset \neq B \subseteq S \}\).
    We start by revealing all adjacencies incident to \(S\).
    The revealed edges partition the vertices of the graph into \(2^{k}\) sets via their connection to \(S\).
    For \(B \subseteq S\), we denote by \(V_B \subseteq [n]\) the vertices whose neighborhood in \(S\) is exactly \(B\).

    Let \(B \subseteq S\).
    Since all edges appear independently with probability \(1/2\),
    the probability that a vertex outside \(S\) is contained in \(V_B\) is \(2^{-k}\).
    Since we only need to specify the asymptotic behavior of the error probability, we can assume at all times \(n\) to be sufficiently large in comparison to \(k\).
    We may assume \(k \le n/2\), which implies \(n-k \ge n/2\), and thus
    the expected number of vertices in this set is
    \[
        E\bigl[|V_B \setminus S|\bigr] = (n-k) \cdot 2^{-k} \ge n / 2^{k+1}.
    \]
    The events that a vertex outside \(S\) is contained in \(V_B\) are all independent.
    The Chernoff bound states that a sum of independent binary random variables significantly differs from its expected value only with exponentially small probability.
    We apply the Chernoff bound and obtain
    \begin{align*}
        \Pr\Bigl[|V_B \setminus S| \le n/2^{k+2} \Bigr] &\le
        \Pr\Bigl[|V_B \setminus S| \le E[|V_B \setminus S|]/2\Bigr] \\ 
        &\le e^{-E[|V_B \setminus S|]/12}\\
        &\le e^{-n / 2^{k+5}}.
    \end{align*}

    By the union bound, the probability that \(E\bigl[|V_B \setminus S|\bigr] \le n/2^{k+3}\) for any \(B \subseteq A\) is at most
    \begin{equation}\label{eq:error1}
        2^k \cdot e^{-n / 2^{k+5}}.
    \end{equation}
    Thus, let us assume (for now) that we revealed the edges incident to \(S\) in such a way that
    \(|V_B \setminus S| \ge n/2^{k+2}\) for all \(B \subseteq S\).
    For sufficiently large \(n\) as a function of \(k\), this implies \(V_B \setminus S\) to be non-empty.
    This lets us pick for every \(\emptyset \neq B \subseteq S\) a ``toggle'' vertex \(v_B \in V_B \setminus S\).

    Consider a vertex \(v \in V_\emptyset \setminus S\).
    We have already revealed the edges between \(v\) and \(S\).
    Let us next reveal the remaining edges between \(v\) and \(V \setminus \bigl( V_\emptyset \cup \{ v_B \mid B \subseteq S \} \bigr) \).
    We define a ``parity pattern''
    \(\mathcal{B}(v) \subseteq \{ B \mid \emptyset \neq B \subseteq S \}\), where \(B \in \mathcal{B}(v)\) if and only if \(v\) is adjacent to an odd number of vertices from \(V_B\).
    The parity pattern \(\mathcal{B}(v)\) will change over time, as we reveal
    one by one the edges between \(v\) and the vertices \(v_B\) with \(\emptyset \neq B \subseteq S\).
    When the edge to \(v_B\) is revealed, the fact whether \(B \in \mathcal{B}(v)\) is inverted with probability \(1/2\), while all other values of the parity pattern stay as they are.
    Since all edges are independent,
    this guarantees that, afterwards, the events whether \(B \in \mathcal{B}(v)\) for \(B \neq \emptyset\) are all independent and appear with probability \(1/2\).
    Since all edges between \(v\) and \(V \setminus V_\emptyset\) were revealed, \(\mathcal{B}(v)\) will not change anymore when further edges are revealed.

    The event that \(\mathcal{B}(v)=\mathcal{B}\) appears with probability exactly \(2^{-2^k-1}\)
    and is equivalent to our goal that
    for all \(\emptyset \neq B \subseteq S\) holds \(B \in \mathcal{B}\) if and only if
    the number of vertices adjacent to every vertex in \(\{v\} \cup B\) and non-adjacent to every vertex in \(S \setminus B\) is odd.
    We proceed revealing edges and analyzing the parity pattern \(\mathcal{B}(v)\) in the same way for all other \(v \in V_{\emptyset} \setminus S\).
    The \(|V_\emptyset \setminus S|\) many events that \(\mathcal{B}(v)=\mathcal{B}\) are all independent.
    Thus, the probability that no vertex \(v \in V_\emptyset\) realizes the parity pattern \(\mathcal{B}(v)=\mathcal{B}\) equals
    (using \Cref{obs:ProbSingleEvent} for the first inequality and \(|V_\emptyset \setminus S| \ge n/2^{k+2}\) for the second inequality)
    \begin{equation}\label{eq:error2}
        \left(1-2^{-2^k+1}\right)^{|V_\emptyset \setminus S|} \le e^{-|V_\emptyset \setminus S|/2^{2^k-1}} \le e^{-n/ \bigl( 2^{k+2} \cdot 2^{2^k-1} \bigr)} = e^{n / 2^{2^{k}+k+1}}.
    \end{equation}
    If there exists a vertex realizing the parity pattern for any choice of \(S\) and \(\mathcal{B}\), the claimed statement is satisfied.
    There are at most \(n^k \cdot 2^{2^{k}}\) choices for \(S\) and \(\mathcal{B}\),
    and for each, the error probability is bounded by the sum of (\ref{eq:error1}) and (\ref{eq:error2}).
    Thus, by the union bound, the total error probability is at most

    \[
        \left(n^k \cdot 2^{2^{k}}\right) \cdot
        \left(2^k \cdot e^{-n / 2^{k+5}} +
        e^{- n / 2^{2^k+k+1}} \right).
    \]
    As \(k\) is fixed, this is a negligible function in \(n\).
\end{proof}

We now return to proving Lemmas \ref{lem:parity2} and \ref{lem:transductionSatisfiesEM}, which concludes the proof of Theorem \ref{thm:ParityGenerator}.

\begin{proof}[Proof of Lemma \ref{lem:parity2}]
    Let us fix \(S \subseteq [n]\).
    For every \(v \in [n]\),
    let \(\mathcal{B}(v)\) be the subset of \(\{ B \mid \emptyset \neq B \subseteq S \}\) such that
    for all \(\emptyset \neq B \subseteq S\) we have \(B \in \mathcal{B}(v)\) if and only if
    the number of vertices that are adjacent to every vertex in \(\{v\} \cup B\) and non-adjacent to every vertex in \(S \setminus B\) is odd.
    Let similarly \(\mathcal{C}(v)\) be the subset of \(\{ C \mid \emptyset \neq C \subseteq S \}\) such that
    for all \(\emptyset \neq C \subseteq S\) we have \(C \in \mathcal{C}(v)\) if and only if
    the number of vertices that are adjacent to every vertex in \(\{v\} \cup C\) is odd.

    For some vertex \(v\) and set \(C\),
    let \(X\) be the vertices that adjacent to every vertex in \(\{v\} \cup C\).
    We may partition the vertices in \(X\) by their adjacencies to \(S \setminus C\).
    Now \(X\) has odd size if and only if we partitioned \(X\) into an odd number of parts of odd size.
    In other words,
    \begin{equation}\label{eq:oddSize}
        C \in \mathcal{C}(v) \iff
        \bigl\{ B \mid C \subseteq B \subseteq S, B \in \mathcal{B}(v) \bigr\} \textnormal{ has odd size}.
    \end{equation}

    One can also show for vertices \(v\) and sets \(B\) that
    \begin{equation}\label{eq:CfromB}
        \mathcal{B}(v) = \{B\}
        \quad\Longrightarrow\quad
        \mathcal{C}(v) = \{C \mid B \subseteq C \subseteq S\}.
    \end{equation}

    In the following, we consider the vector space \(2^{\{ B \mid \emptyset \neq B \subseteq S \}}\) with the symmetric difference operator~\(\oplus\).
    For example, for vectors \(\{B_1, B_2\}\) and \(\{B_2,B_3\}\) of this vector space (with \(\emptyset \neq B_1,B_2,B_3 \subseteq S\)), we have \(\{B_1, B_2\} \oplus \{B_2,B_3\} = \{B_1,B_3\}\).
    We claim that the sets
    \[
        \bigl\{ \{B \mid C \subseteq B \subseteq S\} \mid \emptyset \neq C \subseteq S \bigr\}
    \]
    form a basis of this vector space.
    To prove this, we iteratively construct for decreasing \(i\) all unit vectors \(\{ I\}\) with \(\emptyset \subseteq I \subseteq S, |I| = i\) as linear combinations of our proclaimed basis.
    For \(i>|S|\) there is nothing to show.
    To generate a unit vector \(\{ I\}\) with \(|I|=i-1\),
    we start with the vector \(\{B \mid I \subseteq B \subseteq S\}\)
    and eliminate all excessive entries \(I \subsetneq B \subseteq S\) of size at least \(i\) using the \(\oplus\) operator and the already constructed vectors \(\{B\}\).

    Next, we will argue that for all \(v,v_1,v_2 \in [n]\),
    \begin{equation}\label{eq:commutes}
        \mathcal{B}(v) = \mathcal{B}(v_1) \oplus \mathcal{B}(v_2)
        \quad\Longrightarrow\quad
        \mathcal{C}(v) = \mathcal{C}(v_1) \oplus \mathcal{C}(v_2).
    \end{equation}

    Assume \(\mathcal{B}(v) = \mathcal{B}(v_1) \oplus \mathcal{B}(v_2)\).
    For a set \(C\), we observe
    \begin{align*}
        C \in C(v) &\iff \{ B \mid C \subseteq B \subseteq A, B \in (\mathcal{B}(v_1) \oplus \mathcal{B}(v_2)) \} \txt{ has odd size} \stext{from (\ref{eq:oddSize})}\\
        &\iff \{ B \mid C \subseteq B \subseteq A, B \in \mathcal{B}(v_1) \} \txt{ and } \{ B \mid C \subseteq B \subseteq A, B \in \mathcal{B}(v_2) \} 
        \push\txt{ have sizes of different parity}\\
        &\iff (C \in C(v_1) \txt{ and } C \not\in C(v_2)) \txt{ or } (C \not\in C(v_1) \txt{ and } C \in C(v_2)).
    \end{align*}
    This proves (\ref{eq:commutes}).

    By \Cref{lem:parity1}, we can assume with negligible error probability that,
    for every \[\mathcal{B} \subseteq \{ B \mid \emptyset \neq B \subseteq S \},\] there exists \(v \in [n]\) with \(\mathcal{B}(v) = \mathcal{B}\).
    In particular, for every \(\emptyset \neq B \subseteq S\) there exists \(v \in [n]\) with \(\mathcal{B}(v) = \{B\}\), which, by (\ref{eq:CfromB}),
    implies \(\mathcal{C}(v) = \{C \mid B \subseteq C \subseteq S\}\).
    Hence, by the claim above, \(\{\mathcal{C}(v) \mid v \in [n]\}\) is a basis of our vector space.
    To finally prove the claim of this lemma, pick an arbitrary \(\mathcal{C} \subseteq \{ C \mid \emptyset \neq C \subseteq S \}\).
    Since \(\{\mathcal{C}(v) \mid v \in [n]\}\) is a basis of our vector space,
    we can express \(\mathcal{C}\) as \(\mathcal{C} = \mathcal{C}(v_1) \oplus {} \dots {} \oplus \mathcal{C}(v_t)\) for some vertices \(v_1,\dots,v_t \in [n]\).
    The invocation of \Cref{lem:parity1} at the beginning of the paragraph implies
    the existence of a vertex \(v\) with \(\mathcal{B}(v) = \mathcal{B}(v_1) \oplus {} \dots {} \oplus X(v_t)\).
    Repeated application of (\ref{eq:commutes}) then implies
    \(\mathcal{C}(v) = \mathcal{C}(v_1) \oplus {} \dots {} \oplus C(v_t)\).
    Thus, \(\mathcal{C}(v) = \mathcal{C}\).
\end{proof}

\begin{proof}[Proof of Lemma \ref{lem:transductionSatisfiesEM}]
    By \Cref{lem:parity2},
    with negligible error probability in \(n\),
    for every vertex set \(S \subseteq [n]\) of size \(k\) and every \(\mathcal{C} \subseteq \{ C \mid \emptyset \neq C \subseteq S \}\) there exists \(v \not \in S\) such that
    for all \(\emptyset \neq C \subseteq S\) we have \(C \in \mathcal{C}\) if and only if
    the number of vertices that are adjacent to every vertex in \(\{v\} \cup C\) is odd.
    Thus, in particular, for every \(T \subseteq {S \choose t-1}\)
    there exists \(v \not \in S\) such that
    for every \( \{ s_1,\dots, s_{t-1} \} \in {S \choose t-1}\),
    \(\{ s_1,\dots, s_{t-1} \} \in T\) if and only if the number of vertices that are adjacent to every vertex in \(\{s_1,\dots,s_{t-1},v\}\) is odd.
    The latter is equivalent to \(\{s_1,\dots,s_{t-1},v\}\) being an edge in \(\theta^t(G)\).
    Hence, with negligible error probability, \(G \models \EA^t_k\).
\end{proof}

\section{Conclusion}\label{secConclusion}

In this paper, we define model-theoretic analogues of pseudorandom generators.
We give an exhaustive classification between which
generator and adversary logics from FO, LFP, \LFPpar pseudorandom generators
exist,
and we deterministically construct structures that fool FO and LFP adversaries, even in the realm of arbitrary relational signatures.
Finally, we give a potential candidate for an \FOpar-transduction that acts as a pseudorandom generator that is secure against \FOpar itself. On a conceptual level, this work contributes a novel perspective on pseudorandomness and opens up new avenues for descriptive complexity in the realm of cryptography, where average-case analysis is paramount. 

\bigskip\noindent
\section*{Acknowledgments}
We are very grateful to Boaz Barak for helpful pointers and feedback on an earlier draft.

This material is based upon work supported in part by the National Science Foundation Graduate Research Fellowship Program under Grant No.\xspace DGE1745303. Any opinions, findings, and conclusions or recommendations expressed in this material are those of the authors and do not necessarily reflect the views of the National Science Foundation.

\bibliographystyle{plain}
\bibliography{bibliography}

\clearpage

\appendix

\section{Error Probabilities}\label{apxSimpleFacts}

For completeness, we review a simple fact that is useful for bounding error probabilities.
\begin{observation}\label{obs:ProbSingleEvent}
    Since \(1-x \le e^x\) for all \(x\), we can conclude for all \(a,b\) that
    \[
        \label{eq:inPrelims}
        \Bigl(1-\frac{1}{a}\Bigr)^{b} \le e^{-b/a}.
    \]
\end{observation}
For \(a = 2^{k}\) and \(b = n\), we may bound, for example, the probability of failing to obtain \(k\) times heads in \(k\) coin flips, when being allowed to repeat the experiment \(n\) times, by
\[
    (1-2^{-k})^{n} \le e^{-n/2^k},
\]
which becomes increasingly small if \(n\) is sufficiently large compared to \(k\).
\section{Proof of Lemma \ref{lemRandomCanonization}\label{appRandomCanonizationProof}}

\lemRandomCanonization*
\begin{proof}
    Consider the following algorithm. Let $S$ be the set of elements $x$ of the universe over which $R(x, x, \dots, x)$ holds, and let $\overline{S}$ be its complement. By the Chernoff bound, with probability $1 - \negl(n)$, $S$ and $\overline{S}$ will each contain at least $\frac{n}{3}$ elements total. We define graphs $G$ on $S$ and $\overline{G}$ on $\overline{S}$ where there is an edge between two distinct vertices $x$ and $y$ in the respective vertex sets if and only if exactly one of $R(x, y, y, \dots, y)$ and $R(y, x, x, \dots, x)$ holds. 
    Note that these will both always be uniformly random graphs. 
    We then apply an algorithm of Karp \cite{KarpSplittingCanonization} to both $G$ and $\overline{G}$ to compute canonical total orders on the vertex sets $S$ and $\overline{S}$. 
    Since \(S\) is definable, we can string the orders together with all elements in $S$ occurring before all elements in $\overline{S}$. Hella, Kolaitis, and Luosto \cite{RandomBinaryCanonization} showed that Karp's algorithm can be implemented as an \LFPpar formula, and all the other steps are clearly doable in FO. This gives as an \LFPpar definition of $\leq$. We then define the relation $Q$ to hold on all tuples $(x, \seq{y}{k - 1})$ where $x \in S$ and each $y_i \in \overline{S}$.

    When run on a random graph, the probability that Karp's algorithm fails to find a canonical order is a negligible function of the number of vertices. Thus, assuming each vertex set is of size at least $\frac{n}{3}$, the probability that at least one of the two graphs is not successfully ordered is at most $\negl(\frac{n}{3}) + \negl(\frac{n}{3}) = \negl(n)$. Also, again assuming each vertex set is of size at least $\frac{n}{3}$, $Q$ will contain at least
    $$\left(\frac{n}{3}\right) \cdot \left(\frac{n}{3}\right)^{k - 1} = \frac{1}{3^k} \cdot n^k = \Omega(n^k)$$
    $k$-tuples. Finally, the independence condition follows from the fact that the canonization algorithm only considers the graphs \(G\) and \(\overline{G}\) and thus does not ever check whether any of the tuples in $Q$ are in $R$.
\end{proof}

\section{Relational Rado Structures}\label{apxPseudorandomStructures}

We extend the definition of extension axioms and our deterministic construction of finite Rado graphs to relational structures.

\begin{definition}\label{def:EAstruc}
    Fix a signature \(\sigma\). We say a \emph{\(k\)-atomic type}
    is a set of tuples \((R,(i_1,\dots,i_a))\), where \(R \in \sigma\) is a relation symbol of arity \(a\) and \(i_1,\dots,i_a \in \{0,\dots,k\}\) are indices with \(0 \in \{ i_1,\dots,i_a\}\).
    We denote by \(\types(\sigma,k)\) the set of all possible \(k\)-atomic types.
    If the relations in \(\sigma\) have arities \(a_1,\dots,a_t\), then \(|\types(\sigma,k)| \le 2^{\sum_{i=1}^t (k+1)^{a_i}}\).

    In a \(\sigma\)-structure \(\mathbb A\), we say a vertex \(v \in \mathbb A\) \emph{satisfies} a given \(k\)-atomic type on a tuple \((v_1,\dots,v_k) \in \mathbb A^k\) if
    for all tuples \((R,(i_1,\dots,i_a))\) as above,
    we have \((v_{i_1},\dots,v_{i_a}) \in R^{\mathbb A}\) if and only if \(R(i_1,\dots,i_a)\) is contained in the atomic type.

    We say \(\mathbb A\) satisfies the extension axioms \(\EA^\sigma_s\) if for all pairwise distinct \(v_1,\dots,v_k \in \mathbb A\)
    and all \(k\)-atomic types there exists \(v_0 \neq v_1,\dots, v_k\) satisfying the given type on \((v_1,\dots,v_k)\).
\end{definition}

The construction of finite Rado structures follows along the same lines as \Cref{thm:radograph}, but the additional number of possible atomic types makes
the construction more opaque.
Together with \Cref{lem:EA-probs} and \Cref{lem:EA-determine-qtype}, the following theorem implies (see \Cref{thm:RadoGenerator}) 
that LFP cannot distinguish between random \(\sigma\)-structures and the output of the following theorem.

\RadoStructure*
\begin{proof}
    We instead prove the claim:
    There exists a constant $c$ such that, for every $n,k \in \N$ with $2^{k^c} \le n$,
    one can construct in time $O(2^{2^{ck}} n^c)$ an $n$-element \(\sigma\)-structure satisfying the extension axioms \(\EA^\sigma_k\).
    This claim then implies the statement of the theorem.

    We again use \Cref{lem:tournament} and brute-force enumeration to find
    a tournament $F$ with vertex set \([2^{3k}]\)
    such that, for every $S \subseteq [2^{3k}]$ of size $k$, there is a vertex $j \in [2^{3k}]$
    that has a directed edge towards every vertex in $S$.

    We will construct a structure \(\mathbb A\) with universe \([n]\) satisfying the extension axioms \(\EA^\sigma_k\).
    Let us start by using \Cref{thm:hashfunctions} to construct an $(n,k)$-family of perfect hash functions $\mathcal F$.
    Next, we need to partition the universe \([n]\) into $2^{3k}$ many parts $P_1,\dots,P_{2^{3k}}$ of size at least \(|\mathcal F| \cdot k! \cdot |\types(\sigma,k)|\).
    Let us argue that one can choose the parameter \(c\) such that this is possible for all $n \ge 2^{k^c}$.
    Let $c'$ be the constant from \Cref{thm:hashfunctions}, such that we can bound the size of our $(n,k)$-family of perfect hash functions by $2^{c'k} \log(n)$.
    Let the relations of the arities in \(\sigma\) be \(a_1,\dots,a_t\). Then the number of \(k\)-atomic types is at most \(|\types(\sigma,k)| \le 2^{\sum_{i=1}^t (k+1)^{a_i}}\).
    We may choose \(c\) such that, for every $k \in \N$,
    \[
        2^{3k} \cdot k! \cdot 2^{c'k} \cdot |\types(\sigma,k)| \cdot 2 \le 2^{k^{c}}/k^c.
    \]
    Hence, for every $n \in \N$ with $2^{k^c} \le n$ holds \(2^{k^{c}}/k^c \le n/\log(n)\) and thus
    \[
        2^{3k} \cdot k! \cdot 2^{c'k} \log(n) \cdot |\types(\sigma,k)| \cdot 2 \le n.
    \]
    This means, we can partition the universe into parts \(P_1,\dots,P_{2^{3k}}\), each of size at least
    \[
        \lfloor 2^{c'k} k! \cdot \log(n) \cdot |\types(\sigma,k)| \cdot 2 \rfloor \ge k! \cdot |\mathcal F| \cdot |\types(\sigma,k)|.
    \]

    For convenience, we define an extended set of functions \(\mathcal F^*\) which contains
    for every function \(f \in \mathcal F\) and every permutation \(\pi : [k] \to [k]\) the function \(\pi \circ f\).
    Note that \(|\mathcal F^*| \le |\mathcal F|\cdot k!\) and additionally,
    for all distinct elements \(v_1,\dots,v_k\), there exists \(f \in \mathcal F^*\) with \(f(v_l)=l\) for all \(l \in [k]\).
    We made the parts \(P_1,\dots,P_{2^{3k}}\), large enough so that we can choose a function \(\pattern : [n] \to \mathcal F^* \times \types(\sigma,k)\)
    such that, for each \(j \in [2^{3k}]\), the function \(\pattern\), when restricted to \(P_j\), is surjective.
    We further denote by \(\idx(v)\) the index \(j\) such that \(v \in P_j\).
    To construct our structure \(\mathbb A\),
    we now do the following for every \(R \in \sigma\) or arity \(a\) and sequence \(u_1,\dots,u_a \in \mathbb A\) of (not necessarily distinct) elements:
    \begin{quote}
        Let \(v_0 \in \{u_1,\dots,u_a\}\) be the unique element such that in our tournament \(F\) there are directed edges from \(\idx(v_0)\) to all vertices in \(\{\idx(u_1),\dots,\idx(u_a)\} \setminus \{\idx(v_0)\}\), and let \(\pattern(v_0) = (f,\tau)\).
        Add the tuple \((u_1,\dots,u_a)\) to \(R^{\mathbb A}\) if and only if
        there exists \((R,(i_1,\dots,i_a)) \in \tau\) with
        \[
            \begin{cases}
                u_l = v_0    & \textnormal{for all \(l\) with \(i_l = 0\)},    \\
                f(u_l) = i_l & \textnormal{for all \(l\) with \(i_l \neq 0\)}.
            \end{cases}
        \]
    \end{quote}

    \paragraph*{Correctness}
    Let \(v_1,\dots,v_k \in \mathbb A\) be pairwise distinct and let \(\tau\) be a \(k\)-atomic type.
    We constructed our tournament \(F\) using \Cref{lem:tournament} such that we can choose \(j \in [2^{3k}]\) with a directed edge to all vertices in \(\idx(v_1),\dots,\idx(v_k)\).
    The set \(\mathcal F^*\) was constructed such that there is \(f \in \mathcal F^*\) with \(f(v_l) = l\) for all \(l\).
    Since \(\pattern\), restricted to \(P_j\), was chosen to be surjective, we can also pick \(v_0 \in P_j\) with \(\pattern(v_0) = (f,\tau)\).
    Since \(F\) has no self-loops, \(v_0 \neq v_1,\dots,v_k\).

    Consider now a tuple \((R,(i_1,\dots,i_a))\), where \(R \in \sigma\) is a relation symbol of arity \(a\) and \(i_1,\dots,i_a \in \{0,\dots,k\}\) are indices with \(0 \in \{ i_1,\dots,i_a\}\).
    To prove that \(\mathbb A\) satisfies the extension axioms \(\EA^\sigma_k\),
    we have to show that \((v_{i_1},\dots,v_{i_a}) \in R^{\mathbb A}\) if and only if \(R(i_1,\dots,i_a) \in \tau\).
    Note that \(v_0 \in \{v_{i_1},\dots,v_{i_a}\}\) is the unique vertex such that in \(F\) there is a directed edge from \(\idx(v_0)\) to all vertices in \(\{\idx(u_1),\dots,\idx(u_a)\} \setminus \{\idx(v_0)\}\). Also, \(\pattern(v_0) = (f,\tau)\).
    Additionally, note that \(f\) was chosen such that
    \[
        \begin{cases}
            v_{i_l} = v_0    & \textnormal{for all \(l\) with \(i_l = 0\)},    \\
            f(v_{i_l}) = i_l & \textnormal{for all \(l\) with \(i_l \neq 0\)}.
        \end{cases}
    \]
    Hence, by our construction, \((v_{i_1},\dots,v_{i_a}) \in R^{\mathbb A}\) if and only if \((R,(i_1,\dots,i_a)) \in \tau\), satisfying the extension axiom.

    \paragraph*{Run time}
    The construction of $F$ takes time at most $2^{2^{3k \cdot 2}} \cdot k^2$.
    By \Cref{thm:hashfunctions}, we can construct the $(n,k)$-universal family in time $2^{c'k} n \log(n)$.
    The remaining part of the construction takes time $O(|\types(\sigma,k)|\cdot n^a)$, where \(a\) is the largest arity among relations in~\(\sigma\).
    By possibly rescaling \(c\) such that \(c \ge a\), we get a run time of \(O(2^{2^{ck}} n^c)\).
\end{proof}

\section{Proof of Lemma \ref{lem:EA-probs}}\label{apxEAprob}

\EAProbs*
\begin{proof}
    We prove the statement for relational structures only, as the proof for graphs and hypergraphs is analogous and simpler.
    Fix a signature \(\sigma = \langle R_1,\dots,R_t\rangle\) with relations of arity \(a_1,\dots,a_t\). 
    Then atomic \(k\)-types over \(\sigma\) specify the presence of at most \(\sum_{i=1}^t (k+1)^{a_i}\) connections.
    Fix a tuple \(v_1,\dots,v_k \in \mathbb A^k\) and an atomic \(k\)-type.
    A vertex \(v \not \in S\) satisfies each of the \(\sum_{i=1}^t (k+1)^{a_i}\) connections specified by the \(k\)-type with probability \(1/2\), each.
    There are \(n-k\) choices for \(v\).
    Using \Cref{obs:ProbSingleEvent}, the probability that no vertex satisfies the requirements is
    \[
        \left( 1 - 2^{-{\sum_{i=1}^t (k+1)^{a_i}}} \right)^{n-k} \le e^{(n-k)/{\sum_{i=1}^t (k+1)^{a_i}}}.
    \]
    There are \(n^k\) choices for \(v_1,\dots,v_k\) and at most \(2^{\sum_{i=1}^t (k+1)^{a_i}}\) atomic \(k\)-types. 
    The union bound thus bounds the failure probability over all these options by at most
    \[
        n^k \cdot 2^{\sum_{i=1}^t (k+1)^{a_i}} \cdot 
        e^{(n-k)/{\sum_{i=1}^t (k+1)^{a_i}}}.
    \]
    Since \(\sigma\) and \(k\) are fixed, this function is negligible in \(n\).
\end{proof}

\section{Proof of Lemma \ref{lem:EA-determine-qtype}}\label{appEA-determine-qtype}

\EAType*
\begin{proof}
    We show that two structures satisfy the same LFP sentences up to quantifier rank \(k\) by showing that Duplicator wins
    the infinitary pebble game with \(k\) pebbles~\cite[Definition 11.4]{Libkin2004} on \((\mathbb A_1,\mathbb X)\) and \((\mathbb A_2,\mathbb X)\).
    Assuming we have pebbles \(p_1,\dots,p_k\) to be placed on structure \((\mathbb A_1,\mathbb X)\)
    and pebbles \(q_1,\dots,q_k\) to be placed on \((\mathbb A_2,\mathbb X)\), we need to present a winning strategy
    for Duplicator that always maintains that the two (ordered) \(k\)-tuples of
    pebbles \(p_1,\dots,p_k\) and \(q_1,\dots,q_k\) describe a partial
    isomorphism
    between \((\mathbb A_1,\mathbb X)\) and \((\mathbb A_2,\mathbb X)\).
    If, Spoiler places a pebble \(p_i\) on an \(\mathbb X\)-element of \((\mathbb A_1,\mathbb X)\) 
    then Duplicator places \(q_i\) on the exact same element of \((\mathbb A_2,\mathbb X)\).
    If Spoiler places \(p_i\) on an \(\mathbb A_1\)-element of \((\mathbb A_1,\mathbb X)\),
    then Duplicator observes (in the case of relational structures) the \(\le
    k\)-atomic type that \(p_i\) satisfies in \(\mathbb A_1\) on the subset of
    pebbles of \(p_1,\dots,p_k\) that lie on \(\mathbb A_1\), and uses the
    \(k\)-extension axioms to place \(q_i\) onto an element of \(\mathbb A_2\)
    satisfying the same \(\le k\)-atomic type on the placed pebbles in
    \(\mathbb A_2\). The partial isomorphism of the structures follows from the $k$-atomic types being the same. Thus, Duplicator can follow this strategy ad infinitum and will not lose. The cases of graphs and hypergraphs are completely analogous.
\end{proof}

\section{Proof of Theorem \ref{thmPRGImpossibility}}\label{appPRGImpossibilityProof}

We will need the following lemma, which extends the zero-one law for LFP to formulas with free variables.

\begin{lemma}[Zero-One Law with free variables]\label{lem01LawFormula}
    For every \textnormal{LFP} $\sigma$-formula $\phi(\bx)$ there exists a quantifier free \textnormal{FO} $\sigma$-formula $\overline{\phi}(\bx)$
    such that, as $n$ tends to infinity,
    $$
    \Pr_{\mathbb{A} \sim \struc{\sigma, n}} \Bigl[ \mathbb{A} \models \phi(\bu) \text{ iff } \mathbb{A} \models \overline{\phi}(\bu)
    \text{ for every $|\bx|$-tuple $\bu$ over $[n]$} \Bigr] \ge 1 - \negl(n).
    $$
\end{lemma}

\begin{proof}
    Let $\phi$ be an LFP formula with free variables
    $\seq{x}{\ell}$, of the form $\phi = \exists x_{\ell + 1}\ \phi^*$ where
    $x_{\ell + 1}$ is a new variable and $\phi^*$ is the remainder of the
    formula. 
    Let $\phi^*[x_i / x_{\ell + 1}]$ denote the formula obtained from $\phi^*$
    by replacing all occurrences of the variable $x_{\ell + 1}$ with $x_i$.
    Then we may rewrite $\phi$ equivalently as
	\begin{equation}\label{equBooleanCheckExists}
		\phi = \left(\bigvee_{i \in [\ell]}\phi^*[x_i / x_{\ell + 1}]\right) \vee \exists x_{\ell + 1}\ \left(\left(\bigwedge_{i \in [\ell]} x_{\ell + 1} \neq x_i\right) \wedge \phi^*\right).
	\end{equation}
    Atoms of the form $x_i = x_{l+1}$ and $x_i \neq x_{l+1}$ in $\phi^*$ can be replaced with FALSE and TRUE, respectively.
    Repeating this for for every quantifier, we rewrite $\phi$ such that every quantifier ranges only over variables distinct from $\seq{x}{\ell}$
    and there are no (other) comparisons between free and quantified variables.

    Next, we separate free and quantified variables occuring toegether in relations.
    Let $\sigma'$ be the extended signature, where we add for every $a$-ary relation $R_p$, every $a' \le a$ and every $a'$-tuple $\by$ over the variables from $\bx$
    a new $(a-a')$-ary relation symbol $R^{\by}_i$.
    We obtain a new formula $\phi'$ from $\phi$ by replacing every atom
    $R_i(\by \bz)$, where the variables of the atom are partitioned such that
    $\by$ and $\bz$ contain free and bound variables of $\phi$, respectively,
    with $R^{\by}_i(\bz)$. 
    Since both in $\struc{\sigma, n}$ and $\struc{\sigma', n}$, all relations are distributed independently with probability one half,
    for every tuple~$\bu$
    $$\Pr_{\mathbb{A} \sim \struc{\sigma, n}}[\mathbb{A} \models \phi(\bu)] ~~=~~ \Pr_{\mathbb{A'} \sim \struc{\sigma', n}}[\mathbb{A'} \models \phi'(\bu)].
    $$
    Note that in $\phi'$ variables from $\bx$ occur in no relational atom.
    We next rewrite and partition $\phi'$ in disjunctive normal form
    $$
    \phi(\bx) = \bigvee_{i=1}^s \xi_i(\bx) \land \psi_i(\bx),
    $$
    such that each formula $\xi_i(\bx)$ is quantifier free and each formula $\psi_i(\bx)$
    quantifies only over variables distinct from $\bx$ and contains no relations depending on variables from $\bx$.
    We may therefore see $\psi_i$ as sentences within a $\sigma'$ stucture where the elements $\bx$ have been removed.
    In other words, for an $|\bx|$-tuple $\bu$
    $$
    \Pr_{\mathbb{A} \sim \struc{\sigma, n}}[\mathbb{A} \models \phi(\bu)] ~~=~~ 
    \Pr_{\mathbb{A'} \sim \struc{\sigma', n}}
    ~
    \Pr_{\mathbb{A''} \sim \struc{\sigma'', n-|\bx|}}\left[\bigvee_{i=1}^s \mathbb{A'} \models \xi_i(\bu) \land \mathbb{A''} \models \psi_i\right].
    $$
    By treating each $\psi_i$ as a sentence, the LFP zero-one law of Blass, Gurevich, and Kozen \cite{LFP01Law}
    (or the combination of \Cref{lem:EA-probs} and \Cref{lem:EA-determine-qtype})
    implies that $\psi_i$ is either satisfied with probability
    $\negl(n)$ or $1 - \negl(n)$. Let $\overline{\psi}_i$ be either TRUE or FALSE, accordingly.
    Then for a fixed $\bu$ and $i$, $\overline{\psi}_i$ and $\psi_i(\bu)$ differ with probability $\negl(n)$.
    We have to bound the probability that for \emph{some} tuple $\bu$ and \emph{some} $i$,
    $\overline{\psi}_i$ and $\psi_i(\bu)$ differ.
    There are only a polynomial number of tuples $\bu$ and only a constant number of $i$,
    and thus by the union bound and the fact that the sum of a polynomial number of negligible
    functions is still negligible, the statement holds.
\end{proof}

We are now ready to prove our impossibility result for LFP.

\thmPRGImpossibility*
\begin{proof}~\\
    \ref{itmPRGImpossibilityStatistical}$\implies$\ref{itmPRGImpossibilityNormal}: Obvious.
    
    \ref{itmPRGImpossibilityFormula}$\implies$\ref{itmPRGImpossibilityStatistical}: For positive integers $a$ and $k$, let $\eqrel{a,k}$ denote the set of equivalence relations on $[a]$ with exactly $k$ equivalence classes, and let
        \begin{align*}
            \mathcal{S}_k := \{ (i,\eqrel{[a_i], k}) \mid i \in [t]\}, && \mathcal{S}'_k := \{ (i',\eqrel{[a'_{i'}], k}) \mid i' \in [t'] \}.
        \end{align*}
        
    Note that, for any $a$ and $k$, $T(a, k) = \abs{\eqrel{a, k}} \cdot k!$,
    since choosing a surjective function from a set $S$ to a set of size $k$
    can be thought of as first choosing an equivalence relation determining
    which elements of $S$ will get mapped to the same element, and then
    choosing a bijection from the equivalence classes (of which there must be
    exactly $k$) to the target set of size $k$. Therefore, assuming
    $\sigma \geqS \tau$, we know that, for each positive integer $k$,
    $$\abs{\mathcal{S}_k} = \sum_{i = 1}^t \frac{T(a_i, k)}{k!} \geq \sum_{i = 1}^{t'} \frac{T(a'_i, k)}{k!} = \abs{\mathcal{S}'_k}.$$
    Hence, there exists injections $f_k \colon \mathcal{S}'_k \to \mathcal{S}_k$.

    We describe a transduction $\Theta = (\seq{\theta}{t'}): \struc{\sigma} \to \struc{\tau}$ via the following algorithm. Suppose we wish to determine the truth of $\theta_{i'}(\seq{x}{a'_{i'}})$ for some index $1 \leq i' \leq t'$. We first check equality between the $x$ variables to yield an equivalence relation $\sim'$ on the index set $[a'_{i'}]$. It will have some number $k$ of equivalence classes. Then let $(i, \sim) := f_k(i', \sim')$. Finally, return the truth value of $R_i(\seq{y}{a_i})$, where $y_j$ refers to the variable $x_{j'}$ such that the equivalence relation determined by testing equality among $(\seq{y}{j})$ is $\sim$ and $y_j$ belongs to the $\ell\tth$ equivalence class of $\sim$ where $\ell$ is such that $x_{j'}$ belongs to the $\ell\tth$ equivalence class of $\sim'$. Note that all the ``computation'' here is determining how to route variables to a relation in $\sigma$ solely based on equality comparisons, and can be accomplished purely syntactically by a large (but finite) quantifier-free FO sentence.

    The point of this construction is that, for every $i$ and every $\seq{x}{a_{i'}}$, the truth of $\theta_i(\seq{x}{a_i})$ depends on an independent evaluation of a relation from the input $\sigma$-structure $\mathbb{A}$. This can be verified by checking that no two tuples $(\seq{x}{a_{i'}})$ and $(\seq{x}{a_{j'}})$ are mapped to the same evaluation of relations in $\mathbb{A}$. To see this, observe that, since all variables are used, if $\{\seq{x}{a_{i'}}\} \neq \{\seq{x}{a_{j'}}\}$ then the tuples are mapped to different evaluations in $\mathbb{A}$. Otherwise, if they both contain the same set of $k$ distinct elements but have different equivalence relation types, they are mapped to evaluations of either different relations in $\sigma$ or different equivalence relation types, as $f_k$ is an injection. And if they have the same equivalence relation types, then they are mapped to the same evaluation if and only if the first occurrences are in the same order, i.e., the tuples are literally identical. Thus, by independence $\Theta$ produces uniformly random $\tau$ structures from uniformly random $\sigma$-structures.
    
    $\neg$\ref{itmPRGImpossibilityFormula}$\implies$$\neg$\ref{itmPRGImpossibilityNormal}: Let $\Theta = (\seq{\theta}{t'})$ be given, where each $\theta_i$ is a
    $\sigma$-formula with $a'_i$ free variables.
    By applying \Cref{lem01LawFormula} to each formula of $\Theta$, we obtain a quantifier free first order transduction 
    $\overline{\Theta} := (\seq{\overline{\theta}}{t'})$ such that,
    for $\mathbb{A} \sim \struc{\sigma, n}$,
    the structure $\Theta(\mathbb{A})$ is equivalent to $\overline{\Theta}(\mathbb{A})$ with probability $1 - \negl(n)$.
    
    Since $\overline{\Theta}$ is quantifier-free, it satisfies the hypothesis of Lemma \ref{lemKType} for any $k$. Hence, as we are assuming $\sigma \not\geqS \tau$, Lemma \ref{lemKType} tells us that$\overline{\Theta}$ is not a pseudorandom generator, meaning that there is some sentence $\phi$ that can distinguish its outputs from uniformly random $\tau$ structures with non-negligible advantage. It follows that $\Theta$ cannot be a pseudorandom generator either, as the triangle inequality and the fact that a non-negligible function plus a negligible function is a non-negligible function imply $\phi$ breaks the security of $\Theta$ as well.
\end{proof}

\section{Proof of Lemma \ref{lemKType}}\label{appKTypeProof}

\lemKType*
\begin{proof}
    Let $k$ be a positive integer such that
    \begin{equation}\label{equKViolatesSurjectiveOrder}
        \sum_{i = 1}^{t} T(a_i, k) < \sum_{i = 1}^{t'} T(a'_i, k).
    \end{equation}
    Let $\mathcal{T}_c$ be the set of $(c, k)$-types. For any $c$ and any $T \in \mathcal{T}_c$, let $\phi_T(\seq{y}{c})$ be the FO formula checking that $(\seq{y}{c})$ are all distinct and have $(c, k)$-type $T$. We claim that there exists a $c$ such that, for all sufficiently large $n$, the sentence FO sentence
    $$\phi_c := \bigwedge_{T \in \mathcal{T}_c} \exists \seq{y}{c} \phi_T(\seq{y}{c})$$
    discriminates between $\mathbb{B} \sim \struc{\tau}$ and $f(\mathbb{A})$ for $\mathbb{A} \sim \struc{\sigma, n}$.

    Clearly, when $\mathbb{B} \sim \struc{\tau}$, the probability that $\mathbb{B} \models \phi_c$ is $1 - \negl(n)$; this is a basic consequence of \Cref{obs:ProbSingleEvent}. On the other hand, we will show that, for large enough $c$, there is no $\mathbb{A} \in \struc{\sigma, n}$ such that $f(\mathbb{A}) \models \phi_c$.
    
    We proceed by counting the numbers $N_\sigma$ and $N_\tau$ of possible $(c, k)$-types that a given $c$-tuple could have in any $\sigma$-structure or $\tau$-structure, respectively. Letting $\mathcal{P}(R)$ denote the power set of $R$, starting from the expression in Definition \ref{defCKType}, we have
    \begin{align*}
        N_\sigma &= \abs{\mathcal{P}\left(\bigcup_{i = 1}^t \bigcup_{\substack{S \subseteq [c]\\\abs{S} \leq k}} \{(i, S, g): \suchthat g: [a_i] \to S \txt{ is surjective}\}\right)}\\
        &= \abs{\mathcal{P}\left(\bigcup_{k' = 1}^k \bigcup_{i = 1}^t \bigcup_{\substack{S \subseteq [c]\\\abs{S} = k'}} \{(i, S, g): \suchthat g: [a_i] \to S \txt{ is surjective}\}\right)}\\
        &= 2^{\sum_{k'=1}^k \sum_{i=1}^t T(a_i,k') {c \choose k'}}.
    \end{align*}
    By the assumption on $f$, for any fixed structure $\mathbb{A}$, any given $c$-tuple can only have one of $N_\sigma$ possible $(c, k)$-types in $f(\mathbb{A})$. However, there are
    $$N_\tau = 2^{\sum_{k'=1}^k \sum_{i=1}^t T(a'_i,k') {c \choose k'}}$$
    $(c, k)$-types total. For large enough $c$, the ${c \choose k'}$-terms dominate in the sums over $k'$; for $N_\sigma$ the coefficient of this term is
    $$\sum_{i = 1}^t T(a_i, k),$$
    and for $N_\tau$ the coefficient is
    $$\sum_{i = 1}^t T(a'_i, k).$$
    It thus follows from Equation (\ref{equKViolatesSurjectiveOrder}) that $N_\sigma < N_\tau$. Hence, no matter what $\mathbb{A}$ is, there must be some $k$-type that isn't held by any $c$-tuple in $f(\mathbb{A})$, so $f(\mathbb{A}) \not\models \phi_c$.
\end{proof}

\end{document}